\documentclass[final]{IEEEtran}
\usepackage{cite}
\usepackage{caption}
\usepackage[cmex10]{amsmath}
\usepackage{mathtools}
\usepackage{amsthm}
\usepackage{amsfonts}
\usepackage{amssymb}
\usepackage{stmaryrd}
\usepackage{algcompatible}
\usepackage{blkarray}
\usepackage{url}
\usepackage{graphicx}
\usepackage{bm}
\usepackage{multicol}
\usepackage{float}
\usepackage{subfig}
\usepackage{color}
\usepackage[noend,ruled,norelsize]{algorithm2e}
\usepackage[dvipsnames]{xcolor}

\makeatletter
\def\BState{\State\hskip-\ALG@thistlm}
\makeatother

\newtheorem{definition}{Definition}

\newtheorem{Theorem}{Theorem}
\newtheorem{lemma}{Lemma}

\newcommand{\snr}{\mathrm{SNR}}
\newcommand{\sinr}{\mathrm{SINR}}
\newcommand{\Prob}{\mathrm{Prob}}
\newcommand{\drop}{\mathrm{drop}}

\newcommand{\Var}{\operatorname{Var}}
\newcommand{\tr}{\operatorname{tr}}
\newcommand{\LASC}{LYRRC~}
\newcommand{\LASCNoSpace}{LYRRC}

\makeatletter
\newcommand{\removelatexerror}{\let\@latex@error\@gobble}
\makeatother

\begin{document}
\title{
Balancing Queueing and Retransmission: Latency-Optimal Massive MIMO Design
}
\author{
\IEEEauthorblockN{Xu Du, Yin Sun, Ness B. Shroff, Ashutosh Sabharwal}
\thanks{
Xu Du, Ashutosh Sabharwal are with the Department of Electrical and Computer Engineering, Rice University, Houston, TX, 77005 (e-mails: xdurice@gmail.com, ashu@rice.edu).
Yin Sun is with the Department of Electrical and Computer Engineering, Auburn University (email: yzs0078@auburn.edu).
Ness B. Shroff is with the Departments of ECE and CSE at The Ohio State University (email:shroff@ece.rice.edu).
This work has been supported in part by National Science Foundation awards CCF-1813078, CNS-1518916, CNS-1314822, CNS-1618566, CNS-1719371, CNS-1409336, and from the Office of Naval Research award N00014-17-1-2417.}
}

\bibliographystyle{IEEEbib}
\maketitle
\IEEEpeerreviewmaketitle
\begin{abstract}
  One fundamental challenge in 5G URLLC is how to optimize massive MIMO systems for achieving low latency and high reliability. A natural design choice to maximize reliability and minimize retransmission is to select the lowest allowed target error rate. However, the overall latency is the sum of queueing latency and retransmission latency, hence choosing the lowest target error rate does not always minimize the overall latency. In this paper, we minimize the overall latency by jointly designing the target error rate and transmission rate adaptation, which leads to a fundamental tradeoff point between queueing and retransmission latency. This design problem can be formulated as a Markov decision process, which is theoretically optimal, but its complexity is prohibitively high for real-system deployments. We managed to develop a low-complexity closed-form policy named Large-arraY Reliability and Rate Control (LYRRC), which is proven to be asymptotically latency-optimal as the number of antennas increases. In LYRRC,  the transmission rate is twice of the arrival rate, and the target error rate is a function of the antenna number, arrival rate, and channel estimation error. With simulated and measured channels, our evaluations find LYRRC satisfies the latency and reliability requirements of URLLC in all the tested scenarios.
\end{abstract}


\section{Introduction}\label{sec:intro}
Next-generation cellular systems, labeled as 5G, are targeting low latency and ultra-high reliability to support new forms of applications, e.g. mission critical communications. One of the key technologies for 5G will be massive MIMO, where the base-stations will be equipped with tens to hundreds of antennas~\cite{marzetta2010noncooperative,shepard2012argos,larsson2014massive,CaireL_lower_bound}.
In this paper, we explore how to leverage the large number of spatial degrees of freedom to minimize latency while ensuring high reliability.

Current cellular system design follows a layered approach. The queueing latency\footnote{
In this paper, we use queueing latency to represent the waiting time that packets spend in the MAC-layer queue.
And overall latency denotes the total latency caused by retransmission and waiting at the MAC-layer queue.
}
is managed at MAC and higher layers, while the target (block) error rate\footnote{
In this paper, we use the target error rate when emphasizing the design of transmission control.
And we use block error rate when emphasizing the probability of decoding error under a given transmission control.}
is managed separately by the physical layer to maximize the physical layer throughput.
For example, the transmission rate (usually referred to as modulation and coding scheme~\cite{3gpp.36.213}) is often adapted to meet a fixed target error rate of around $10$\%.
This decoupled design is shown to be nearly throughput optimal~\cite{wu2011coding} for single-antenna systems. However, such a decoupled design may not achieve low latency.

As 5G pushes to low latency (10-100$\times$ lower than the LTE system~\cite{LTE_measurment_2017}) and ultra-high reliability, it is of paramount importance to control the latency and service unreliability caused by retransmissions.
The Ultra-Reliable Low-Latency Communication (URLLC) has a reliability requirement of $99.9999$\%~\cite{7980747}, i.e., the probability of packet successful delivery within $4$ round of transmissions ($0.25$ ms$/$5G frame) should be higher than $99.9999$\%.
To satisfy such reliability requirement, the target error rate cannot exceed $3.16$\%.
For a given set of possible target error rates, it might be natural to choose the lowest one, which leads to the highest link reliability and shortest retransmission latency.
However, since the overall latency is the sum of latency due to queueing and due to retransmissions, a very small target error rate might result in long queueing latency and does not always minimize the overall latency.
In this paper, we achieve reliability guaranteed latency minimization by finding the target error rate
and the transmission rate adaptation that jointly minimize the overall latency.

\begin{figure*}[htbp]
  \centering
  \includegraphics[width=0.35\textwidth]{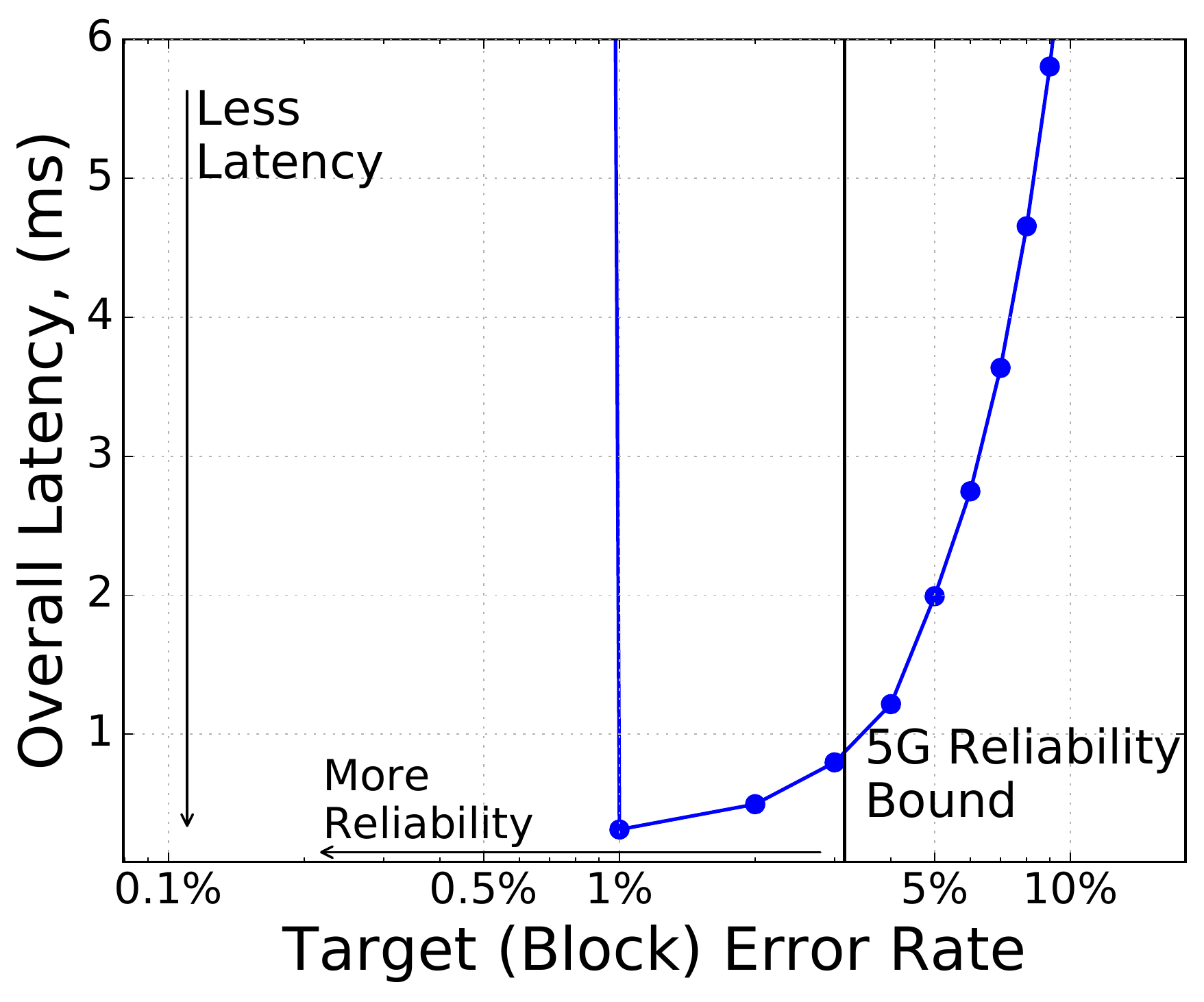}
  \caption{
An example illustrating the overall latency for different target error rates, where the transmission rate has been optimized for each given target error rate. A massive MIMO uplink system with $4$ single-antenna users and $32$ base-station antennas is considered. The channel traces are measured in an over-the-air channel on the Rice Argos platform and the base-station estimates the channel based on $8$ pilot symbols per user.
Please find the evaluation details in Section~\ref{sec:numerical}.
  }
  \label{fig:tradeoff}
\end{figure*}

While it is widely known that the target error rate reduces with a higher transmission power or a lower transmission rate, the relationship between the target error rate and overall latency is more complex.
There is a tradeoff between retransmission latency and queueing latency, both of which are impacted by the target error rate:
On the one hand, the retransmission latency reduces as the target error rate reduces.
On the other hand, if the system is fixed to an extremely low target error rate, few packets can be transmitted in each frame, i.e., the transmission time to send the same amount of packets increases, and packets have to wait for a longer time in the queue.
Therefore, under a given arrival process, the queueing latency increases as the target error rate reduces.
The situation is further complicated by the fact that current mobile users adapt their transmission power, which makes the feasible (transmission rate, target error rate) tuple time-varying.
Fig.~\ref{fig:tradeoff} depicts an example of the minimum overall latency achieved at different target error rates where the transmission rate is optimized for given target error rate; the details on how to optimize the transmission rate will be discussed later in Section~\ref{sec:algo}.
For the specific example in Fig.~\ref{fig:tradeoff}, a target error rate (1\%) smaller than both the LTE target error rate (10\%) and the URLLC reliability requirement (target error rate of 3.16\%) results in the minimum overall latency.
It demonstrates a need for finding an appropriate target error rate that minimizes the overall latency by balancing the queueing latency with the retransmission latency.

In this paper, we model practical massive MIMO systems with retransmissions.
To minimize the overall latency from both queueing and retransmission, we optimize the target error rate and transmission rate adaptation.
The main contributions of this paper are the following:
\begin{itemize}
\item
We formulate a latency minimization problem for massive MIMO systems, in which the target error rate and transmission rate are jointly optimized for minimizing the overall latency, subject to the reliability constraint of URLLC.
The arrival process is a discrete random process that is memoryless.
This optimization problem is cast as a constrained Markov decision process and solved by value iteration.

\item Because Markov decision process does not provide much insight on the optimal control, we develop a deterministic control policy for massive MIMO with a large number of antennas and a constant arrival rate.
We note that there exists an important 5G URLLC type data traffic, e.g., time-sensitive and throughput-hungry virtual reality (VR) service~\cite{7946930}, which has a constant data arrival rate.
This deterministic control policy is named as Large-arraY Reliability and Rate Control (LYRRC), which has a low complexity and is in a closed form: If the packet arrival rate is $\lambda$, the transmission rate of LYRRC is $2\lambda$.
In addition, the target error rate of LYRRC is $F_{\eta}\left[\frac{1}{M^{1-\rho}} \left(1+\frac{K}{\tau}+p_{I}\right)\right]$, where $F_{\eta}$ is the CDF of the effective channel gain (defined later), $M$ is the number of base-station antennas, $K$ is the number of users, $\rho$ is the traffic arrival load over link capacity, $p_{I}$ is the power of the interference from neighboring cells, and $\tau$ is the number of pilots.
LYRRC is proven to be asymptotically optimal as the number of antennas grows to infinity. Furthermore, the total latency achieved by LYRRC can be expressed as a closed-form function of the number of base-station antennas $M$, the number of pilots $\tau$, the number of served users $K$, and $\rho$.
In particular, for $\rho\in\left[0,1\right)$,
we show that the average waiting time diminishes to zero as
$M$
 increases to infinity.

\item To verify \LASCNoSpace's performance in the real world, we measure massive MIMO channels on the $2.4$ GHz with Rice Argos platform~\cite{shepard2012argos}, which consists of a $64$-antenna base-station and four mobile users.
The numerical experiments based on the measured and simulated channels show that \LASC with 5G self-contained frame~\cite{QC.S_CA.Patent,3gpp.36.213} can simultaneously meet the $1$ ms latency and $99.9999$\% reliability criterion.
In the same scenario, the best latency of transmission rate control policies  with a fixed target error rate of $10\%$ is more than $5$ ms.
The evaluations demonstrate that \LASC can provide $400\times$ latency reduction compared to current LTE transmission control, which has a target error rate of $10\%$ and fixed per-frame transmission power control.
Compared to the best queue-length based rate adaptation policy with a fixed target error rate of $10\%$, \LASC achieves a $20 \times$ latency reduction.
\end{itemize}

{\sl Related Work:}
The majority of the massive MIMO literature focuses on the achievable rate maximization, which assumes full-buffer and does not model the upper layer latency from queueing.
Massive MIMO was shown to provide higher spectral efficiency~\cite{marzetta2016fundamentals,bjornson2017massive}, wider coverage~\cite{marzetta2016fundamentals,bjornson2017massive} and easier network interference management~\cite{6815892, ngo2013energy, marzetta2016fundamentals} than traditional MIMO.
This work differs from previous massive MIMO physical layer work in that we provide reliability guaranteed latency-optimal transmission control.
Prior work also optimized the retransmission process, either for throughput~\cite{wu2011coding} or energy efficiency~\cite{su2011optimal} maximization.
Additionally, cross-layer optimization~\cite{lin2006tutorial, 1413227, 1235598, makki2014noisy} have been proposed for latency reduction.
For a point-to-point system, past studies~\cite{berry2002communication, 4294166, 4567575, cao2008power} showed that using the queue-length information for transmission rate control can reduce queueing latency.
Finally, stochastic network calculus~\cite{8362864} is used to capture the latency violation probability of multi-input single-output systems with perfect rate adaptation.
Thus, the perfect rate adaptation of past work implies no decoding error or retransmission latency.

The remainder of this paper is structured as follows. In Section~\ref{sec:SysMdl}, we provide a physical layer abstraction and network model for a single user latency minimization problem.
Section~\ref{sec:algo} provides an algorithm to solve the formulated latency minimization problem. A simple and yet latency-optimal transmission control policy, \LASCNoSpace, is investigated in the large-array regime in Section~\ref{sec:large_M}.
In Section~\ref{sec:MU}, we extend our single-user analytical results to multiuser massive MIMO systems.
We provide numerical results in Section~\ref{sec:numerical} and conclude in Section~\ref{sec:Conclude}.

{\sl Notations:} We use boldface to denote vectors/matrices.
We use $|\cdot|$ to denote the magnitude of a complex number. And the $l_2$ norm of a complex vector is $\left\|\cdot\right\|$.
The complex space is $\mathbb{C}$.
The space of real value is $\mathbb{R}$ whose positive half is denoted as $\mathbb{R}^{+}$. The following notations are used to compare two non-negative real-valued sequences $\left\{a_{n}\right\}$, $\left\{b_{n}\right\}$:
$a_{n}=O\left(b_{n}\right)$ if $\lim_{n \to \infty} \frac{a_{n}}{b_{n}} \leq \infty$;
$a_{n}=o\left(b_{n}\right)$ if $\lim_{n \to \infty} \frac{a_{n}}{b_{n}} =0$.
And $f_{1}\left(M\right) \cong f_2\left(M\right)$ denotes that $\lim_{M \to \infty} \frac{f_1\left(M\right)}{f_2\left(M\right)} = 1$.

\section{System Model and Problem Formulation}\label{sec:SysMdl}
\subsection{System Model}
We consider a massive MIMO uplink system. The single-user case is considered first in Sections II-IV, and is depicted in Fig.~\ref{fig:simo-mdl}. The extension to multi-user systems will be presented later in Section~\ref{sec:MU}. Each user is equipped with a single antenna and the base station has M antennas.
\begin{figure*}[htbp]
  \centering
  \includegraphics[width=0.8\textwidth]{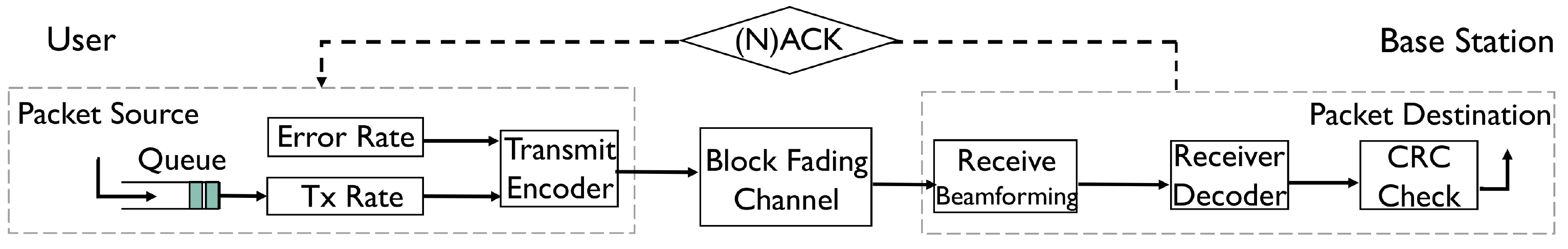}
  \caption{Single-user uplink system consisting of a single antenna user and an $M$-antenna base-station.}
  \label{fig:simo-mdl}
\end{figure*}
Based on the physical layer procedures defined in the first 5G release~\cite{3gpp.36.213}, we consider that the system operates in self-contained frames, as shown in Fig.~\ref{fig:self_frm}.
A self-contained frame consists of both data transmission and an immediate ACK/NACK.
 Without loss of generality, the duration of each frame is of  $1$ unit and Frame $t$ spans the time interval $\left[t, t + 1 \right), \ t \geq 0$.
In each frame, the user first transmits encoded data packets to the base-station. The base-station then feeds back an ACK or NACK to signal whether a decoding error occurred.
The feedback is assumed to be error free.

\begin{figure*}[htbp]
  \centering
  \includegraphics[width=0.75\textwidth]{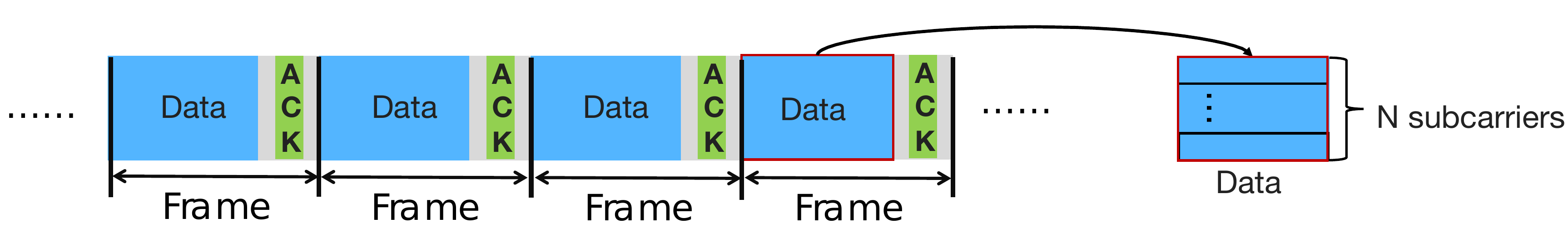}
  \caption{Structure of the self-contained frames. Each self-contained frame consists of uplink data resource blocks (blue), downlink feedback signals (green) and the guard periods (gray). The transmitted data is encoded over $N$ subcarriers with a single code-block.
 }~\label{fig:self_frm}
\end{figure*}

\subsubsection{Physical Layer Model}~\label{subsec:sysmdl-phy}
During the uplink data transmission, the received signal
 by the base-station over the wideband channel
 is
\begin{equation}
\mathbf{y}_{n}= \sqrt{\gamma} \mathbf{h}_{n} x_{n} + \mathbf{z}_{n}, \quad n = 1, ..., N, \label{equ:dl_sig}
\end{equation}
where $n$ is the subcarrier index,
$N$ is the total number of subcarriers,
$x_{n}$ is the transmitted signal,
$\mathbf{z}_{n} \in \mathbb{C}^{M}$ is a zero-mean circularly symmetric complex Gaussian noise vector,
and $0<\gamma\leq 1$ is the large-scale channel gain.
We model the channel fading processes as block Rayleigh fading, where the small-scale fading vector $\mathbf{h}_{t,n}$ maintain the same during each frame and varies independently across frames and subcarriers.
In this paper, we may omit the frame index $t$ in $\mathbf{h}_{t,n}$ when the frame index is clear from the context.
During each frame, the user transmits $\tau$ uplink pilots, each with power $p_{\tau}$. Let $\hat{\mathbf{h}}_{n}$ be the estimated channel vector by the base-station via the MMSE estimator. The estimated channel satisfies that~\cite{marzetta2016fundamentals, bjornson2017massive}
\begin{equation}
\mathbf{h}_{n} = \hat{\mathbf{h}}_{n} + \mathbf{e}_{n},
\end{equation}
where $\mathbf{e}_{n} \in \mathbb{C}^{M}$ is a zero-mean, circularly symmetric complex Gaussian noise vector with variance of $\frac{1}{1+\gamma p_{\tau} \tau}$. After applying conjugate beamforming, the obtained signal is
\begin{align}
\hat{x}_{n} &=  \hat{\mathbf{h}}_{n}^{H} \mathbf{y}_{n}
=   \hat{\mathbf{h}}_{n}^{H} \left[\sqrt{\gamma} \left(\hat{\mathbf{h}}_{n} + \mathbf{e}_{n}\right) x_{n} + \mathbf{z}_{n}\right] \notag\\
&= \sqrt{\gamma}\hat{\mathbf{h}}_{n}^{H} \hat{\mathbf{h}}_{n} x_{n} + \sqrt{\gamma} \hat{\mathbf{h}}_{n}^{H} \mathbf{e}_{n} x_{n} + \hat{\mathbf{h}}_{n}^{H} \mathbf{z}_{n}, \label{equ:y_n_su}
\end{align}
where the three terms on the right hand side represent the desired signal, signal loss from imperfect channel knowledge, and noise, respectively.
The receive $\sinr$ on Subcarrier $n$ is~\cite{ngo2013energy, 6415388}
\begin{equation}
\sinr_{n} = \frac{\gamma p}{ \frac{\gamma p}{1+ \gamma p_{\tau} \tau} + 1} \left\|\hat{\mathbf{h}}_{n}\right\|^{2},~\label{equ:SU_SINR}
\end{equation}
where $p = |x_{n}|^2$ is the power of uplink data transmission.

The user is aware of the large-scale channel gain $\gamma$ and the distribution of the small-scale channel fading via the estimation of a periodic indication signal broadcast by the base-station~\cite{3gpp.36.213}.
During each frame, all uplink packets to be transmitted are encoded in a single code block that spans all $N$ subcarriers.
The block error rate of the uplink transmission $\epsilon$ is a function of the transmission power.
A closed-form characterization of the block error rate appears to be intractable when the code-block length is finite~\cite{5452208}.
Hence, we employ the following block error rate approximation that was developed in~\cite{wu2010performance,wu2011coding,5452208,8640115,8529234}.
Let $L$ be the number of information bits in each packet, and $r_{t}$ is the number of transmitted packets in Frame $t$. We refer to $r_t$ as the {\sl transmission rate}.
The block error rate of a code block with a code-block length $L_{code}$ can be approximated as
\begin{align}
\epsilon
\approx & \Prob\left[\sum^{N}_{n=1}
\log\left(1+\sinr_{n} \right) - \frac{\nu}{
\sqrt{
L_{\mathrm{code}}
}
}\leq r L\right] \label{equ:p_to_epsilon1}\\
\approx &
\Prob\left[\sum^{N}_{n=1}\log\left(\sinr_{n}\right) \leq r L\right],
~\label{equ:p_to_epsilon}
\end{align}
where $\nu$ is the channel dispersion~\cite{5452208,8640115} due to finite block length and is upper bounded by $\log_{2}\left(e\right)$.
For a systems with strong channel coding,~\cite{5452208} shows that~\eqref{equ:p_to_epsilon1} closely captures the block error rate when $L_{\mathrm{code}}>100$.
The approximation in~\eqref{equ:p_to_epsilon} is derived by considering sufficiently large code-block length~\cite{wu2010performance,8529234,wu2011coding} and high $\sinr$ regime~\cite{wu2010performance,wu2011coding}.
Fig.~\ref{fig:outage_channel} provides an illustration of the approximated block error rate in~\eqref{equ:p_to_epsilon}, in which an LDPC-based massive MIMO system is considered and the code-block length is chosen according to DVB-S.2 standard.
Our simulations confirm the conclusions drawn from past works~\cite{wu2010performance,wu2011coding,8529234}.
We hence adopt\footnote{
One can also use the block error rate approximation~\eqref{equ:p_to_epsilon1} which is more accurate in the low SINR and short code-block length regime.
In this case, the effective channel gain in~\eqref{equ:eta_def} and power mapping in~\eqref{equ:p_per_frm} should be modified accordingly.
}~\eqref{equ:p_to_epsilon}
as the block error rate model.

\subsubsection{Buffer Dynamics with Retransmission}\label{subsec:su_buff_dyna}
We assume that there is no packet in the buffer at time $0$.
During each frame, $\lambda$ new packets arrive in the queue\footnote{
Our model and analysis can be directly generalized to the case where the number of new arrival packets across frames follow an independent and identically distribution.
} and each packet contains $L$-bits.
In each frame, the user receives downlink ACK/NACK feedback from the base-station.
Upon ACK, the transmitted packets are removed from the buffer. Upon NACK, the transmitted packets remain at the buffer queue head\footnote{
It is possible to reduce the power of retransmissions via the joint decoding of failed packets and retransmissions as in HARQ. For mathematical tractability, we consider that the receiver discards undecoded packets.}.
We use the indicator function $1_{t}$ to represent decoding success, $1_{t} = 1$ means success and $1_{t} = 0$ otherwise. The distribution of the $1_{t}$ is determined by the chosen target error rate $\epsilon$ as $P\left[1_{t}=1\right] = 1 - \epsilon$ and $P\left[1_{t}=1\right] = \epsilon$.

At time $t$, let $q_{t}$ be the queue-length of the buffer, and $r_{t}$ be the number of packets to be transmitted at Frame $t$ as per the control decision. The queue-length evolves according to
\begin{equation}
q_{t+1} = \min\left[\max\left(q_{t} + \lambda - 1_{t}r_{t}, \lambda\right), B\right],
~\label{equ:buff_evlo}
\end{equation}
where $B$ is the size of the buffer and $r_t$ is the number of transmitted packets in Frame $t$.
If the buffer cannot store all the packets waiting to be transmitted, an overflow event occurs.
The number of dropped packets due to the buffer overflow is given by
\begin{equation}
b_{t} = \max\left(q_{t} + \lambda - 1_{t}r_{t} - B, \lambda - B\right).~\label{equ:overflow}
\end{equation}
The average number of dropped packets due to overflow, measured in packets per frame, is $\lambda_{\drop} = \lim_{T  \to \infty} \sum_{t=0}^{T-1} b_{t}/T$.
When packet overflow happens, the dropped packets induce significant latency to time-sensitive applications.
We assume that each overflowed packet introduces a large latency penalty $D_{\drop}$.
We are interested in minimizing the overall latency (from arrival to successfully delivery).
We consider the stationary policies are complete, i.e., the minimum latency can be achieved by a stationary policy.
Under a stationary policy, the queueing latency of successfully served packets are
$
  \lim_{T \to \infty} \frac{1}{T}\sum_{t=0}^{T-1} \frac{q_{t}}{\lambda - \lambda_{\drop}} \label{equ:queue_cost}
$,
which is derived by using Little's Law~\cite{bertsekas1992data}.
To summarize, if a packet is dropped, its latency is $D_{\drop}$ and if a packet is successfully served (not dropped), its latency is $\lim_{T \to \infty} \frac{1}{T}\sum_{t=0}^{T-1} \frac{q_{t}}{\lambda - \lambda_{\drop}}$.
The average latency is then
\begin{align}
D &= \frac{\lambda-\lambda_{\drop}}{\lambda} \lim_{T \to \infty}  \frac{1}{T} \sum_{t=0}^{T-1} \frac{q_{t}}{\lambda- \lambda_{\drop}} + \frac{\lambda_{\drop}}{\lambda} D_{\drop} \notag\\
&=  \frac{\bar{q}}{\lambda} +  \frac{\lambda_{\drop}}{\lambda} D_{\drop}, ~\label{equ:sys_mdl_d}
\end{align}
where $\frac{\lambda-\lambda_{\drop}}{\lambda}$ is the proportion of successfully served packets and $\bar{q}$ is the average queue-length, i.e., $ \lim_{T \to \infty}  \sum_{t=0}^{T-1} \frac{q_{t}}{T}$.

\begin{figure*}[htbp]
  \centering
  \includegraphics[width=0.4\textwidth]{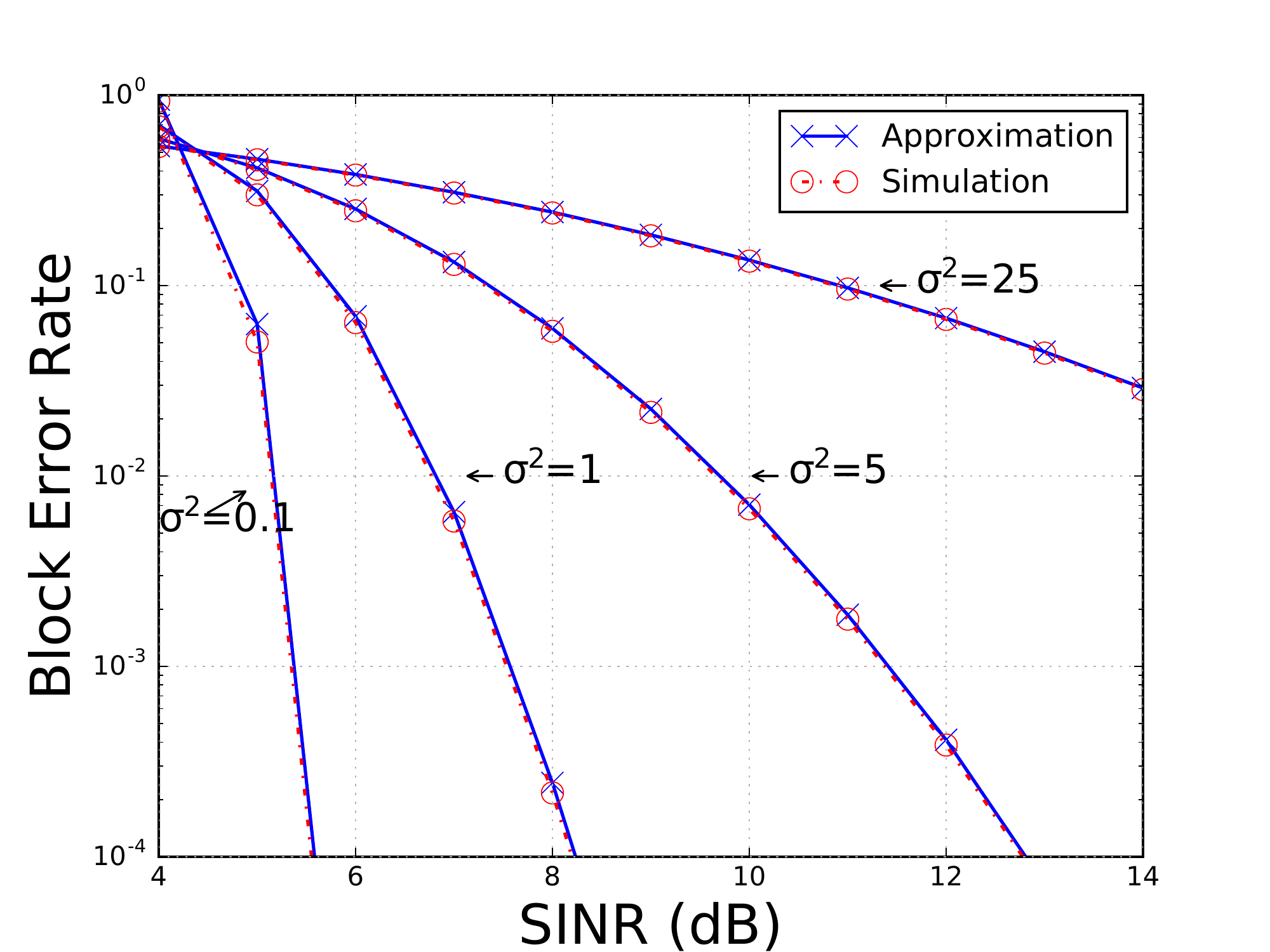}
  \caption{Block error rate of a coded system as a function of $\sinr$ mean with $N=1$.
  In simulation, the channel gain follows the normal distribution with labeled variance.
  The approximations are obtained by~\eqref{equ:p_to_epsilon}.
  And the simulation is done with LDPC code~\cite{MATLAB} and sparse parity-check matrix comes from the DVB-S.2 standard. The transmission is at a rate of $1.5$ bits per symbol ($8$-QAM, $0.5$ code rate).
  }
  \label{fig:outage_channel}
\end{figure*}

\subsubsection{Transmission Power adaptation}
We consider the transmission power of the user to satisfy a long-term power constraint of~$P$.
In Frame $t$, the transmission power is adapted, based on the transmission rate $r_{t}$, and the number of pilots $\tau$, to achieve the target error rate $\epsilon$.
The transmission power is quantified in the sequel:
Substituting~\eqref{equ:SU_SINR} into~\eqref{equ:p_to_epsilon}, the block error rate is approximated as
\begin{equation}
\epsilon \approx \Prob\left[	\left(\prod^{N}_{n=1} \kappa_{n} \right)^{1/N}\leq \frac{\exp\left(r L/N\right)}{M} \left(
\frac{1}{1+\gamma p_{\tau} \tau}
+ \frac{1}{\gamma p}
 \right)
\right].\label{equ:su_outage}
\end{equation}
where $\kappa_{n}$ is the the per-antenna gain of small-scale channel fading, given by
\begin{equation}
\kappa_n \stackrel{\Delta}{=}   \left\|\hat{\mathbf{h}}_{n}\right\|^2/M. \label{equ:kappa_su_def}
\end{equation}
The per-antenna gain $\kappa_{n}$ is the arithmetic mean of the small-scale channel gain across the $M$ antennas because the received signals with different antennas are combined during the linear beamforming.
The left-hand-side of the inequality of~\eqref{equ:su_outage} is determined by the small-scale fading, and the right-hand-side of~\eqref{equ:su_outage} is a constant independent of small-scale fading.
For the ease of subsequent presentation, we define
\begin{equation}
\eta \stackrel{\Delta}{=} \left(\prod^{N}_{n=1} \kappa_{n} \right)^{1/N},
\label{equ:eta_def}
\end{equation}
which is called \emph{effective channel gain}.
The effective channel gain~\eqref{equ:eta_def} is the geometric mean across the $N$ subcarriers because the maximum outage-free rate~\cite{5452208} can be approximated by the logarithmic of the product of the per-subcarrier $\sinr_{n}$.
Let $F_{\eta}\left(x\right) \stackrel{\Delta}{=} \Prob\left(\eta\leq x\right)$ denote the cumulative distribution function (CDF) of the effective channel $\eta$. And the inverse CDF of $\eta$
is~$F_{\eta}^{-1}\left(\epsilon\right) \stackrel{\Delta}{=} \inf \left\{x\in {\mathbb  {R}^{+}}:\epsilon\leq F_{\eta}(x)\right\}$.
Recall that the transmission power is adapted to achieve the target error rate, from~\eqref{equ:su_outage}, we have
\begin{equation}
p\left(r, \epsilon, \tau \right) = \left[\frac{M \gamma F_{\eta}^{-1}\left(\epsilon\right)}{\exp\left(rL/N\right)} - \frac{\gamma}{1 + \gamma p_{\tau} \tau } \right]^{-1},
~\label{equ:p_per_frm}
\end{equation}
where $F_{\eta}^{-1}$ is the inverse CDF of the effective channel gain $\eta$ in~\eqref{equ:eta_def}.
When $\tau$ increases, the base-station has a more accurate channel estimation and the needed transmission power (at the same rate with the same reliability) reduces.
One can observe that the required transmission power increases with the transmission rate $r$ and the packet size $L$, and decreases with the  number of base-station antennas $M$, the number of subcarriers $N$, and the number of pilots $\tau$.

\subsection{Single-user Latency Minimization Problem}
We now formulate the single-user latency minimization problem.
The objective of the joint target error rate and transmission rate control is to minimize the average packet latency under a long-term average power constraint.
The {\sl system state} is the queue-length $q_t$, whose state space is $\mathcal{Q}= \{0, 1, ..., B\}$.
The transmission controller determines the number of transmitted packets $r_t$ at the beginning of each frame based on the queue-length $q_t$, as well as the target error rate $\epsilon$ that remains constant in all frames over time.
Recall that the the transmission rate is the number of transmitted packets $r_{t}$.
We consider the set of stationary policies such that $r_t$ = $\mu (q_t)$, where $\mu:\mathcal{Q}\to \mathbb{R}^{+}$ is a function.
And the target error rate $\epsilon$ is chosen from a finite set $\mathcal{E}$.
Finally, the transmission power $p_{t}$ is adapted based on the designed rate $r_{t}$, target error rate $\epsilon$, and number of pilot $\tau$ as in~\eqref{equ:p_per_frm}.
Both the transmission rate function $\mu$ and the resulting transmission power are independent of the exact small-scale fading $\mathbf{h}_{n}$ as it is unknown to the user.

For any target error rate $\epsilon$ and transmission rate function $\mu$, we assume that the resulted Markov chain of the system states is ergodic, i.e., the unichain condition is satisfied.
The associated unique steady state of the system is denoted as $\pi$.
The latency minimization problem is formulated as:
\begin{subequations}
\begin{align}
  \min_{\substack{
  \epsilon \in \mathcal{E},\\  r_{t} =\mu\left(q_{t}\right) ,\\ \mu:\mathcal{Q}\to \mathbb{R}^{+}}
  }  \quad  & \quad
      D = E\left[\frac{\bar{q}}{\lambda} + \frac{\lambda_{\drop}}{\lambda} D_{\drop}\right] ~\label{equ:latnecy_obj}\\
\text{s.t.} \quad \quad \  & \quad  E\left[\lim_{T\to \infty}\frac{1}{T} \sum_{t=0}^{T-1}
p\left(r_{t}, \epsilon, \tau\right)\right] \leq  P,~\label{equ:power_constraint}\\
& \quad  \epsilon \leq  \epsilon_{\max},~\label{equ:reliable_constraint}\\
& \quad   \text{State Transition Model~(4)-(8)},
\end{align}\label{equ:opt_form}%
\end{subequations}
where $\epsilon_{\max}$ it the maximum allowed target error rate due to reliability requirement. For 5G URLLC, $\epsilon_{\max}=\left(1-99.9999\%\right)^{1/4}=3.16\%$.
The optimal objective value of~\eqref{equ:opt_form} is denoted as $D^{*}$, or $D^{*}\left(M\right)$ when we need to emphasize the dependence on the number of antennas $M$.
Hence, $D^{*}\left(M\right)$ captures the minimum overall latency $D^{*}$ as a function of the number of base-station antennas $M$.


\section{Latency-Optimal Single-User Transmission Control}~\label{sec:algo}
In this section, we first formulate the latency minimization problem~\eqref{equ:opt_form}  as a constrained average cost Markov Decision Process (MDP) and solve it by a proposed algorithm.
The proposed algorithm can also solve the latency-optimal control for general point-to-point MIMO systems by replacing the per-subcarrier $\sinr$ in~\eqref{equ:SU_SINR} with the $\sinr$ of the MIMO system.
The effective channel gain in~\eqref{equ:eta_def} and power mapping in~\eqref{equ:p_per_frm} also should be modified accordingly.

\subsection{Lagrange Duality of the MDP}~\label{subsec:mdp_form}
For a target error rate $\epsilon \in \mathcal{E}$, and a stationary transmission rate adaptation $\mathcal{Q}\to \mathbb{R}^{+} $, based on the definition of average latency~\eqref{equ:sys_mdl_d}, we define the induced latency cost mapping $d$ on each state action pair as
\begin{equation}
d \left(q_{t}, r_{t}, \epsilon\right) = \frac{q_{t}}{\lambda} + \frac{b_{t}}{\lambda} D_{\drop}, \notag
\end{equation}
where $b$ is the number of the dropped packet due to buffer overflow as shown in~\eqref{equ:overflow}.
In Frame $t$, a latency cost and a transmission power cost are incurred.
The average overall latency of the problem in infinite horizon equals
\begin{equation}
D_{\pi} = E_{\pi} \left[\lim_{T \to \infty} \frac{1}{T} \sum_{t=0}^{T-1}d\left(q_{t},r_{t},\epsilon\right) \right]. ~\notag
\end{equation}
Similarly, utilizing the transmission power characterization in~\eqref{equ:p_per_frm}, the average power is
\begin{equation} 
P_{\pi} = E_{\pi} \left[\lim_{T \to \infty} \frac{1}{T} \sum_{t=0}^{T-1}p( r_{t}, \epsilon,\tau)\right]. ~\notag
\end{equation}
Given an average power constraint $P$, the objective of the joint target error rate selection and transmission rate control is restated as a constrained MDP as
\begin{align}
& \text{Minimize $D_{\pi}$}  \notag\\
\text{subject to} \quad &  \quad  P_{\pi} \leq P, \epsilon \leq  \epsilon_{\max}, \notag\\
&  \quad  \text{State Transition Model (4)-(8)}. ~\label{equ:CMDP}
\end{align}
The constrained MDP~\eqref{equ:CMDP} is converted to an unconstrained MDP via Lagrange's relaxation as
\begin{align}
  & \text{Minimize $D_{\pi} + \beta P_{\pi}$}  \notag\\
  \text{subject to} \quad & \quad  \epsilon \leq  \epsilon_{\max}.~\label{equ:UMDP}
\end{align}

For ergodic MDP,~\cite{zero_dual, goyal2008optimal} provide a sufficient condition under which the unconstrained MDP is also optimal for the original constrained problem~\eqref{equ:opt_form}.
For all policies such that $P_{\pi}=P$, the sufficient condition provided by~\cite{zero_dual, goyal2008optimal} is satisfied. Thus, when the constraint is binding, there exists zero-duality gap between original problem~\eqref{equ:opt_form} and the unconstrained MDP~\eqref{equ:UMDP}, i.e., their optimal solution is the same.

We now present the algorithm to solve~\eqref{equ:UMDP} in Section~\ref{subsec:mdp_algo}.
The closed-form solution of~\eqref{equ:UMDP} and the characterization of the array-latency tradeoff $D^{*}\left(M\right)$ are presented in Section~\ref{sec:large_M}.

\subsection{A Value Iteration Based Algorithm}~\label{subsec:mdp_algo}
Problem~\eqref{equ:UMDP} is an MDP with an average cost criterion in infinite horizon.
To find the optimal target error rate, we need to find the optimal transmission rate adaptation and the corresponding achievable latency for each $\epsilon \in \mathcal{E}$ that is smaller than $\epsilon_{\max}$.
Furthermore, for each target error rate $\epsilon$, we can use binary search method to find the smallest $\beta$ that satisfies the long-term power constraint $P$ in~\eqref{equ:UMDP}.
Such $\beta$ corresponds to the latency-optimal solution for~\eqref{equ:CMDP} because that, for each $\epsilon$, the average power is monotonically non-decreasing on $\beta>0$.
Finally, for each $\epsilon$ and $\beta$, we thus find the optimal transmission rate adaptation $\mu^{*}$ by considering
$\alpha$-discounted problem~\cite{bertsekas2007dynamic_bk} of~\eqref{equ:UMDP}.
We now present a solution to each of the discounted problem.
For each system state $q$, define value cost function as
\begin{equation*}
  V_{\alpha} \left(q\right) \stackrel{\Delta}{=} \min_{\mu} E_{\pi} \left\{\sum_{t=0}^{\infty}\alpha^{t}\left[d\left(r_{t}, q_{t}, \epsilon\right) + \beta p\left(r_{t}, \epsilon, \tau\right)\right]\right\},
\end{equation*}
where $\alpha \in \left(0, 1\right)$ is the discount factor. For each $\epsilon$ and $\beta$, we need to find a stationary transmission rate adaptation for all $\alpha$-discounted problem with $\alpha \in \left(0, 1\right)$, i.e., the Blackwell optimal policy. For the considered {\sl finite} state MDP, the Blackwell optimal policy~\cite{bertsekas2007dynamic_bk} exists and is also optimal for the average cost problem~\eqref{equ:UMDP}. The Bellman's equation of the above $\alpha$-discounted problem is then
\begin{align}
V^{*}_{\alpha} \left(q\right) = \min_{\mu}  \Big\{d&(r, q, \epsilon) + \beta p\left(r, \epsilon, \tau \right)  + \notag\\
 \big[&  \left(1-\epsilon\right) V^{*}_{\alpha}\left(\min\left(q + \lambda - r, B\right) \right) + \notag \\
& \epsilon V^{*}_{\alpha}\left(\min\left(q + \lambda, B\right)\right)
\big]\Big\}, ~\label{equ:Bellman}
\end{align}
whose state transition is described by~\eqref{equ:p_to_epsilon},~\eqref{equ:buff_evlo}, and~\eqref{equ:overflow}. Using dynamic programming with value iteration~\cite{bertsekas2007dynamic_bk} over~\eqref{equ:Bellman}, we can solve the $\alpha$-discounted problem. Since the discounted cost $V_{\alpha}$ is bounded,~\cite{bertsekas2007dynamic_bk} shows that solving~\eqref{equ:Bellman} generates the optimal transmission rate control $\mu^{*}$.

We summarize the above steps in Algorithm~\ref{alg:dp},
\begin{figure}[!t]
 \removelatexerror
 \setlength{\belowcaptionskip}{-30pt}
  \begin{algorithm}[H]~\label{alg}
    \caption{ Latency-Optimal Joint Target Error Rate and Transmission Rate Control}\label{alg:dp}
    \SetKwInOut{Input}{Input}
    \SetKwInOut{Output}{Output}
    \Input{Average power constraint $P$, number of antennas $M$, number of subcarriers $N$, distribution of packet arrival $a$, large-scale channel gain $\gamma$, CDF of effective channel gain $\eta$, number of pilots $\tau$, pilots power $p_{\tau}$.}
    \Output{Optimal target error rate $\epsilon^{*}$, optimal transmission rate adaptation $\mu^{*}$, minimum achievable latency $D^{*}$.}
   \For(\quad $\sslash$ \emph{Find minimum latency for each $\epsilon \in \mathcal{E}$}){$\epsilon \in \mathcal{E}$ that $\epsilon \leq \epsilon_{\max}$}
   {
      $\beta_{min} = 0, \ \beta_{max} = z$; \quad $\sslash$ \emph{$z$ is a very large but finite number} \\
      \While ( \quad $\sslash$ \emph{Find smallest $\beta$ that satisfies the average power constraint, $\delta$ is a small constant that controls the algorithm output accuracy}){$\beta_{\min} /\beta_{\max} < 1- \delta $}{
        $\beta \leftarrow \left(\beta_{\max} + \beta_{\min}\right)/2$ \;
        Initialize $V_{\alpha}^{0}\left(q\right)$ for every system state in $\mathcal{Q}$ and $n=1$\;
        Solve for $V_{\alpha}^{1}$ from $V_{\alpha}^{0}$ via value iteration as~\eqref{equ:Bellman}\;
        \While (\quad $\sslash$ \emph{Find optimal $\mu$ for each $\beta$ and $\epsilon$}) {$V_{\alpha}^{n} \neq V_{\alpha}^{n-1}$}{
          Update $V_{\alpha}^{n}$ from $V_{\alpha}^{n-1}$ via value iteration as~\eqref{equ:Bellman}\;
        }
        Compute the corresponding power $P_{\mathrm{tmp}}$\;
        \eIf{$P_{\mathrm{tmp}} > P$} { $\beta_{\min}$ = $\beta$\;}{$\beta_{\max}$ = $\beta$\;}
      }
      Denote the solved transmission rate function as $\mu_{\epsilon}\left(q_{t}\right)$ and the resulted latency as $D_{\epsilon}$.
   }
   {\bf Optimal policy extraction:} $\epsilon^{*}= \arg\min_{\epsilon\in\mathcal{E}, \epsilon\leq\epsilon_{\max}} D_{\epsilon}$,
      $\mu^{*}\left(q_{t}\right) = \mu_{\epsilon^{*}}\left(q_{t}\right)$, and $D^{*}=D_{\epsilon^{*}}$.
  \end{algorithm}
\end{figure}
which solves~\eqref{equ:CMDP} to find the optimal target error rate and transmission rate adaptation.
To provide insights on the structure of optimal transmission controls, we now resent a closed-form characterizations when $M\to\infty$ in Section~\ref{sec:large_M}.

\section{Large-Array Latency-Optimal Control}\label{sec:large_M}
In this section, we derive the latency-optimal control for the single-user problem in~\eqref{equ:opt_form} when the number of base-station antennas $M\to \infty$.
For the single-user system in Rayleigh fading, the per-antenna gain $\kappa_{n}$ in~\eqref{equ:kappa_su_def} satisfies the following~\cite[A.2.4]{marzetta2016fundamentals},\cite{ngo2013energy, bjornson2017massive}.
\begin{itemize}
\item {\sl Mean}: The per-antenna gain mean is a constant that is independent of $M$, i.e.,
\begin{equation}
  E\left[\kappa_{n}\right] = \frac{\tau p_{\tau} \gamma }{\tau p_{\tau} \gamma + 1},\label{equ:admissive_mean}
\end{equation}
\item {\sl Variance}: The per-antenna gain variance is inversely proportional to $M$, i.e.,
\begin{equation}
\Var\left[\kappa_{n}\right]=\frac{1}{M}\left(\frac{\tau p_{\tau} \gamma }{\tau p_{\tau} \gamma + 1}\right)^2. \label{equ:admissive_variance}
\end{equation}
\end{itemize}
In Section~\ref{sec:MU}, we will show that a multiuser massive MIMO channel can be decoupled into parallel single-user channels. For each of the decoupled channels, the per-antenna gain is also of variance that is inversely proportional to $M$.

Based on condition~\eqref{equ:admissive_mean}, the achievable $\sinr$ grows with the number of base-station antennas $M$ linearly.
As the focus of the current section is on the asymptotic analysis with $M \to \infty$, we can view $\log M$ as the link ``capacity''. In the same spirit, we define the system utilization factor to be a constant $\rho\in\left[0,1\right)$ as
\begin{equation}
\rho \stackrel{\Delta}{=} \frac{\lambda L}{N \log M }, \label{equ:rho_def}
\end{equation}
where $\lambda$ is the packet arrival rate, $L$ is the number of bits in each packet, and $N$ is the number of subcarriers. By~\eqref{equ:rho_def}, the packet arrival rate $\lambda$ increases with $M$ and equals $\frac{N \log M}{L \rho}$.
Conceptually, the term $N \log M$ can be viewed as the total ``capacity'' of the wideband link and $\lambda L$ can be viewed as the data load.
Thus, the utilization factor $\rho$ can be interpreted as the ratio between the offered data load and the total link ``capacity''.

We also make the following assumptions for mathematical tractability. We consider an infinite buffer (i.e., $B \rightarrow \infty$), thus no buffer overflow or overflow latency occurs. And the target error rate $\epsilon$ can be chosen from a continuous set $\left(0, 1\right)$.

\subsection{Array-Latency Scaling Lower Bound}~\label{subsec:large_array_lower_bd}
Notice that a trivial lower bound of $D^{*}(M)$ is $1$ frame, which is the first transmission attempt of a packet. This $1$ frame latency lower bound can only be achieved if the target error rate is exactly zero.
We now provide a tighter lower bound of the array-latency curve $D^{*}\left(M\right)$.

\begin{Theorem}[Latency Scaling Lower Bound]~\label{thm:large_array_lower_bd}
 The optimum array-latency curve $D^{*}\left(M\right)$ satisfies
  \begin{equation}
  D^{*}\left(M\right) - 1\geq
  \frac{\epsilon_{o}}{1-\epsilon_{o}},~\label{equ:D_star_M}
  \end{equation}
  where $\epsilon_{o}$ is given by
  \begin{equation}
    \epsilon_{o} = F_{\eta}\left[ \frac{1}{M^{\left(1 - \rho\right)}} \left(\frac{1}{\gamma P} + \frac{1}{\gamma p_{\tau} \tau}\right)\right], \label{equ:eps_oo}
  \end{equation}
  where  $F_{\eta}\left(\cdot\right)$ is the
  CDF of the effective channel gain $\eta$ in~\eqref{equ:eta_def}, $\rho \in \left[0, 1\right)$ is the utilization factor in~\eqref{equ:rho_def}, and $\tau$ is the number of pilots.
\end{Theorem}

\begin{proof}
  The main idea is to lower bound the overall latency by the packet retransmission latency, which monotonically increases with the target error rate.
  To complete the proof, we use Jensen's inequality to show that there exists a minimum target error rate $\epsilon_{o}$ such that for any $\epsilon<\epsilon_{o}$ the long-term throughput is smaller than $\lambda$.
  Appendix~\ref{appendix:delay_gap} provides the proof details.
\end{proof}

Theorem~\ref{thm:large_array_lower_bd} presents a latency lower bound.
For any transmission rate adaptation, $\epsilon_{o}$ is the minimum target error rate that leads to a long-term throughput no smaller than $\lambda$.
And if the target error rate is smaller than $\epsilon_{o}$, the queue-length process will not stable.
By the definition of $\eta$~\eqref{equ:eta_def}, the per-antenna mean~\eqref{equ:admissive_mean}, and the per-antenna variance~\eqref{equ:admissive_variance}, Chebyshev's inequality can be used to show that $\epsilon_{o}$ converges (in probability) to $0$ as the number of base-station antenna $M$ increases to infinity.
The channel hardening effect can explain such convergence. The latency lower bound~\eqref{equ:D_star_M} hence converges to $0$ as $M\to\infty$.

If $\tau p_\tau$ is small, the channel estimation error is large.
As a result, both $\epsilon_{o}$ and the latency lower bound are large.
In this case, neither high reliability nor low latency can be met. Hence, sufficiently good channel estimation is necessary for achieving high reliability and low latency.

\subsection{Large-Array Optimal Target Error Rate and Transmission Rate Control}~\label{subsec:opt_ctrl_lrg_array}
In this subsection, we present a simple transmission control policy that meets with the latency lower bound in (20) asymptotically as $M \to \infty$.

\begin{definition}
We define the {\sl Large-arraY Reliability and Rate Control}~(\LASCNoSpace) as
\begin{equation}
  \begin{cases}
    \epsilon^{*} = \epsilon_{o}\\
    \mu^{*}: r_{t}\left(q_{t}\right) = \min\left(q_{t}, \ 2 \lambda  \right)
  \end{cases},~\label{equ:eps_l_mu_l_def}
\end{equation}
where $\epsilon_{o}$ is given by~\eqref{equ:eps_oo}.
\end{definition}
The \LASC policy contains two parts: a target error rate of $\epsilon_{o}$ and an transmission rate control policy $\mu^{*}$.
The transmission rate adaptation $\mu^{*}$ describes a simple thresholding rule: If there are more than $2\lambda$ packets in the buffer queue, i.e., $q \geq 2\lambda$, $2\lambda$ packets will be transmitted. If less than $2\lambda$ packets are currently in the buffer, all packet in the queue will be scheduled for transmission in the frame.
In each frame, based on the transmission rate of $\min\left(q_{t},2\lambda\right)$, the user utilizes power adaptation~\eqref{equ:p_per_frm} to achieve the target error rate target $\epsilon_{o}$.

To evaluate \LASCNoSpace, we now first derive the latency with arbitrary target error rate $\epsilon<\frac{1}{2}$ and transmission rate policy $\mu^{*}$.
We next prove the asymptotic optimality of \LASC\eqref{equ:eps_l_mu_l_def} by comparing the achieved latency to the minimum latency lower bound in Theorem~\ref{thm:large_array_lower_bd}.

\subsubsection{Latency Performance of Transmission Rate adaptation $\mu^{*}$}
\begin{lemma}
Under any target error rate $\epsilon<\frac{1}{2}$ and transmission rate adaptation $r_{t}\left(q_{t}\right)= \min\left(q_{t}, \ 2 \lambda \right)$,
the overall latency is $1 + \frac{\epsilon}{1-2\epsilon}$.~\label{lemma:tx_rate}
\end{lemma}
\begin{proof}
The main idea is to compute the steady state distribution of the queue-length, which is a Markov chain with infinite countable states. Appendix~\ref{appendix:tx_rate} provides the complete proof.
\end{proof}

Lemma~\ref{lemma:tx_rate} provides a closed-form characterization of the transmission rate adaptation $\mu^{*}$ when the maximum buffer-length is infinite.
To provide insights on the proof of Lemma~\ref{lemma:tx_rate}, we consider the associated Markov chain of the buffer-length. The buffer-length state transition under any target error rate $\epsilon \in \left(0, 1\right)$, which is not necessarily equal to $\epsilon_{o}$, and the transmission rate adaptation $\mu^{*}$ is depicted in Fig.~\ref{fig:q_mu_l}.
By Little's Law, the overall latency equals to the ratio between the average queue-length and the arrival rate $\lambda$. Notice that $\lambda$ is the difference between the adjacent states in Fig.~\ref{fig:q_mu_l}. Hence, the average queue-length is in proportional with $\lambda$ (see Appendix~\ref{appendix:tx_rate} for a rigorous proof). As a result, the overall latency depends only on the target error rate $\epsilon$, but not on $\lambda$.

\begin{figure*}[ht]
\centering
\includegraphics[width=0.7\textwidth]{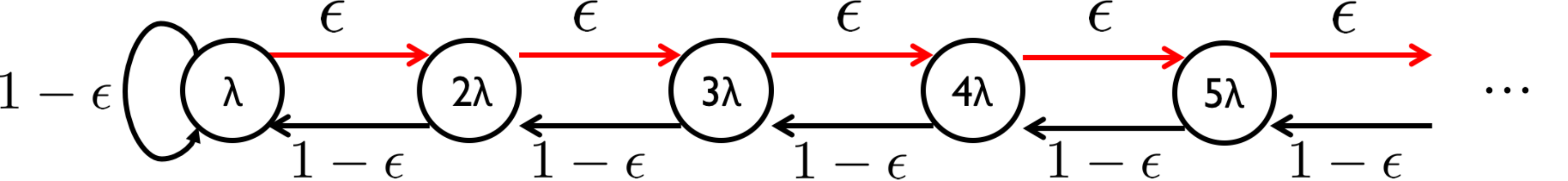}
\caption{Evolution of the queue-length $q_t$ under any target error rate $\epsilon \in \left(0, 1\right)$ and the transmission rate adaptation $\mu^{*}$ as a Markov chain.
If $\epsilon>0.5$, the average queue-length hence queueing latency is infinite.}
~\label{fig:q_mu_l}
\end{figure*}

To summarize, the transmission rate control policy $\mu^{*}$ applies a negative drift $-\lambda$ with probability $(1-2\epsilon)$ towards the minimum queue-length $\lambda$.
To minimize the latency as $M\to \infty$, the queue-length needs to be regulated towards the minimum queue-length $\lambda$. This regulation is achieved by selecting a smaller target error rate.

By using Lemma~\ref{lemma:tx_rate}, we have that the achieved latency of \LASC is
$D_{\mathrm{LYRRC}}\left(M\right) = 1 + \frac{\epsilon_{o}}{1-2\epsilon_{o}}$.
As mentioned above, the target error rate $\epsilon_{o}$ of \LASC\eqref{equ:eps_l_mu_l_def} reduces as the number of base-station antennas increases.
The achieved latency $D_{\mathrm{LYRRC}}$ reduces with more base-station antennas.
We now prove the asymptotic optimality of \LASCNoSpace.


\subsubsection{Asymptotic Optimality of \LASCNoSpace}
\begin{Theorem}[Optimal Large-Array Control] \label{thm:large_array_contrl}
For any $\rho \in \left[0,1\right)$ and positive $\tau$, as $M \to \infty$, \LASC\eqref{equ:eps_l_mu_l_def} guarantees that the overall latency is within a vanishing gap from optimal as
\begin{equation}
    D_{\mathrm{LYRRC}}\left(M\right) - D^{*}\left(M\right) \cong \left(\epsilon_{o}\right)^2, \ M\to\infty,
\end{equation}
where $D_{\mathrm{LYRRC}}\left(M\right)= 1 + \frac{\epsilon_{o}}{1-2\epsilon_{o}}$ is the overall latency by \LASCNoSpace, and $\epsilon_{o}$ is given by~\eqref{equ:eps_oo}.
\end{Theorem}
\begin{proof}
We first characterize the gap between latency under \LASC and minimum latency by combining Lemma~\ref{lemma:tx_rate} and Theorem~\ref{thm:large_array_lower_bd}. The proof is complete by using the large deviation theory to show that the power constraint is satisfied. Please see Appendix~\ref{appendix:large_array_contrl} for details.
\end{proof}

Recall that $f_{1}\left(M\right) \cong f_2\left(M\right)$ denotes that $\lim_{M \to \infty} \frac{f_1\left(M\right)}{f_2\left(M\right)} = 1$.
Theorem~\ref{thm:large_array_contrl} establishes the asymptotic optimality of \LASCNoSpace.
In addition, the latency gap between the lower bound and \LASC increases as the channel estimation error increases ($\tau$ reduces).
Furthermore, Lemma~\ref{lemma:tx_rate} and Theorem~\ref{thm:large_array_contrl} suggest that the latency-optimal target error rate increases for systems with fewer base-station antennas.
Hence, the reliability and low-latency design objectives of 5G URLLC does not always matches with each other for practical massive MIMO system with finite $M$.
Finally, we note that \LASC can achieve optimal-latency for any $\rho \in \left[0,1\right)$, which seems to contradict the transmission rate of $\min\left(q_{t}, 2\lambda\right)$. This can be explained by the fact that we are considering a wireless link with power adaptation and the probability of transmit at $2\lambda$ reduces as $M\to\infty$.
Therefore, using larger transmission power (over a few frames) can increase the peak transmission rate beyond the long-term average rate.
We next combine Theorem~\ref{thm:large_array_contrl} and Theorem~\ref{thm:large_array_lower_bd} to characterize the scaling of the array-latency curve $D^{*}\left(M\right)$ in closed-form.

\begin{Theorem}[Large-Array Latency Scaling]
  As $M\to \infty$, for any positive $\tau$ and $\rho\in\left[0,1\right)$, the optimum latency converges to $1$ frame as
  \begin{equation}
      D^{*}\left(M\right) - 1\cong \epsilon_{o} , \ M \to \infty
  \end{equation}
  where $F_{\eta}\left(\cdot\right)$ is
  the CDF function of the effective channel gain $\eta$, and $\epsilon_{o}$ is given by~\eqref{equ:eps_oo}.
  ~\label{thm:large_array_scaling}
\end{Theorem}
\begin{proof}
Theorem~\ref{thm:large_array_lower_bd} provides a latency lower bound. The optimal joint control in Theorem~\ref{thm:large_array_contrl} serves as an achievability proof and provides an upper bound. The proof is complete by showing that the ratio of the upper bound and the lower bound converges to $1$ as $M\to \infty$.
\end{proof}

Theorem~\ref{thm:large_array_scaling} provides a closed-form characterization of the large-array latency.
In {\sl closed-form}, it describes the minimum latency $D^{*}$ as a function of the utilization factor $\rho$, the channel estimation error, and the number of base-station antennas $M$. As $M \to \infty$, $\epsilon_{o} \to 0$. Thus, both the retransmission and queueing latency converges to $0$ frame.
Finally, we comment on the impact of imperfect channel state information.
For any $\tau>0$, the latency convergence to the $1$ frame as $M\to\infty$.
For a practical system with finite $M$, more accurate channel leads to smaller latency.

\section{Multi-user Extension}\label{sec:MU}
In this section, we now consider the $K$-user latency minimization problem over the lossy channel.
In this section, suffix $\left[k\right], \ k= 1, 2, \cdots, K$ denotes the user index.
The long-term power constraint of User $k$ is $P\left[k\right]$.
The multiuser controller decides the target error rate $\epsilon\left[k\right]$ and the transmission rate $r_{t}\left[k\right]$ of User $k$. The buffer dynamic of each user is identical to that of the single user counterpart that is described in Section~\ref{subsec:su_buff_dyna}.

To minimize the system latency of the $K$ users at the same time, we associate positive weights~$\omega_{k}, \ k=1, \ldots, K$ to users. The multiuser latency minimization problem is then
\begin{equation}
\begin{split}
\min_{
\substack{
\epsilon\left[k\right],\ r_{t}\left[k\right] \\ \forall k}
}  & \quad
      \sum_{k=1}^{K} \omega_{k} D\left[k\right]\\
\text{s.t.}   & \quad E\left[\lim_{T\to \infty}\frac{1}{T} \sum_{t=1}^{T} p_{t}\left[k\right] \leq  P\left[k\right]\right], \ \forall k, \\
&\quad \epsilon\left[k\right] = \Prob\left[\sum^{N}_{n}\log\left(\sinr_{t,n}\left[k\right]\right) \leq r_{t}\left[k\right] L\right], \ \forall k,\\
& \quad  \epsilon\left[k\right] \leq  \epsilon_{\max}\left[k\right], \ \forall k,
\end{split}~\label{equ:mu_coupled}
\end{equation}
where $\epsilon_{\max}\left[k\right]$ is the maximum allowed target error rate (minimum reliability) of User $k$.
And $\sinr_{t,n}\left[k\right]$ is the receiver $\sinr$ of the $n$-th subcarrier in Frame $t$ for User $k$. Here, the buffer length $q_{t}\left[k\right]$ and buffer overflow $b_{t}\left[k\right]$ of User $k$ is given by~\eqref{equ:buff_evlo} and~\eqref{equ:overflow}, respectively.

To detect signals from the $K$ users, the base-station applies receive beamforming.
Let matrix $\mathbf{H}_{n} \in \mathbb{C}^{M \times K}$ denotes the uplink small-scale channel fading between the $M$-antenna base-station and the $K$ users.
Throughout this section, we consider user channels follow i.i.d. Rayleigh fading.
Finally, the base-station receives an inter-cell interference that is modeled by an additive white Gaussian noise of power~$p_{I}$, which is independent of the estimated channel.

Let the estimated channel and estimation error be $\hat{\mathbf{H}}_n$ and $\tilde{\mathbf{H}}_n$, respectively.
With the MMSE estimator, the estimation error between each base-station antenna and User $k$ is an complex Gaussian random variable with zero mean and variance of $\frac{1}{\tau p_{\tau}\left[k\right] \gamma\left[k\right]+1}$. Here, $\tau$ and $p_{\tau}\left[k\right]$ are the number of uplink pilots and the pilot power, respectively.
The base-station use the estimated channel to generate zero-forcing receive beamformers to detect the uplink signal of each user.
The receive beamforming matrix is $\mathbf{V}_{n} \stackrel{\Delta}{=} \left(\hat{\mathbf{H}}_{n}^{H}\hat{\mathbf{H}}_{n}\right)^{-1}\hat{\mathbf{H}}_{n}^{H}$.
On Subcarrier $n$, the received signal of User $k$ is~\cite{marzetta2016fundamentals,bjornson2017massive}
\begin{equation}
  \hat{x}_{k} =
  \sqrt{p\left[k\right]\gamma\left[k\right]} x_{k} +
  \left[\left(\mathbf{H}^{H}\mathbf{H}\right)^{-1}\hat{\mathbf{H}}_{n}^{H}\left(\mathbf{z} + \mathbf{z}_{I}-  \tilde{\mathbf{H}}\mathbf{x}
  \right)\right]_{K},
\end{equation}
where $\mathbf{z}$ and $\mathbf{z}_{I}$ are the receiver noise and inter-cell interference, respectively.
Similarly to past work~\cite{8529234,8640115} on retransmission, we compute the $\sinr$ by treating the interference as the worst case Gaussian noise.
And the effective SINR for User $k$ on Subcarrier $n$ is
\begin{align}
  &\sinr_{n} \left[k\right] \notag\\
  &= \frac{p_{k}\gamma_{k}}{\left(1+
p_{I} +
  \sum_{i=1}^{K}\frac{p\left[i\right]\gamma\left[i\right]}{\tau p_{\tau}\left[i\right]\gamma\left[i\right] + 1}\right) \left[\left(\hat{\mathbf{H}}_{n}^{H}\hat{\mathbf{H}}_{n}\right)^{-1}\right]_{kk}},
  \label{equ:new_mu_sinr}\end{align}
where $\left[\cdot\right]_{kk}$ denotes the $k$-th diagonal element of a matrix.
A crucial property of the $\sinr_{n}$ term~\eqref{equ:new_mu_sinr} is that the randomness of both the channel variation and the interference is concisely described by the inverse of the estimated channel, which is a random matrix.

For a practical uplink system where each user is unaware of other users' channel or queue information, the joint target error rate and transmission rate adaptation design appears intractable.
To see the difficulty of the joint policy design, we consider the following example.
For each user, the inter-beam interference in~\eqref{equ:new_mu_sinr} depends on {\sl other} users' large-scale fading and transmission power.
Recall that each user's transmission power changes in each frame based on its current queue-length.
Thus, it is extremely difficult for each user with only local knowledge (queue-length and large-scale fading) to infer the exact value of $\sum_{i=1}^{K}\frac{p\left[i\right]\gamma\left[i\right]}{\tau p_{\tau}\left[i\right]\gamma\left[i\right] + 1}$ and hence the proper transmission power. As a result, the target error rate and transmission rate policy cannot be designed distributedly by each user, which is undesirable for a practical uplink system.

Here, we proceed with the observation that, in real-world systems, the pilot power is usually required to be higher than the data signal power~\cite{3gpp.36.213}.
Hence, the $\sum_{i=1}^{K}\frac{p\left[i\right]\gamma\left[i\right]}{\tau p_{\tau}\left[i\right]\gamma\left[i\right] + 1}$ term is upper bounded by $\frac{K}{\tau}$, which can be viewed as a {\sl worst} cast interference penalty.
Each user then adjusts its power based on the $\sinr$ loss upper bound.
Substituting the $\sinr$ expression~\eqref{equ:new_mu_sinr} of the multiuser system into~\eqref{equ:p_to_epsilon}, we then have that the target error rate as
\begin{equation}
\epsilon \approx \Prob\left[
\left(\prod^{N}_{n=1}
\kappa_{n}
\right)^{1/N}
\leq \left(1+\frac{K}{\tau} + p_{I}\right) \frac{\exp\left(rL/N\right)}{M p \gamma}\right],
\label{equ:outage_MU}
\end{equation}
where the per-antenna gain $\kappa_{n}$ is
\begin{equation}
\kappa_{n}=\left\{ M\left[\left(\hat{\mathbf{H}}_{n}^{H}\hat{\mathbf{H}}_{n}\right)^{-1}\right]_{kk}\right\}^{-1}.~\label{equ:kappa_mu}
\end{equation}
Similarly to the single-user case, we also compute the per-frame transmission power as
\begin{equation}
p\left(r, \epsilon, \tau \right) = \left(1+\frac{K}{\tau} + p_{I}\right) \frac{\exp\left(rL/N\right)}{ F_{\eta}^{-1}\left(\epsilon\right)M \gamma},~\label{equ:MU_power_Map}
\end{equation}
where $\epsilon$ is the scheduled reliability target (target error rate) and $r$ is the transmission rate (in unit of packet).
Here, $\approx$ in~\eqref{equ:outage_MU} is because that each user considers the upper bound of inter-beam interference.

The per-antenna gain~\eqref{equ:kappa_mu} is independent of the large-scale channel, transmission power, and hence queue-length of the other $K-1$ users.
For each user, the distribution of the effective channel $\eta$ in~\eqref{equ:eta_def} then becomes independent of the channel, queue-length, and power of the other users.
Therefore, we can decouple the multiuser problem. By adopting a new distribution of the effective channel gain $\eta$ (generated by~\eqref{equ:kappa_mu}) and the new power mapping~\eqref{equ:MU_power_Map}, the multiuser problem is decoupled to $K$ independent single user problems~\eqref{equ:opt_form}. Each of the single-user problems can be solved by Algorithm~\ref{alg:dp}.
We now further demonstrate that the large-array analytical results in Section~\ref{sec:large_M} also apply to the considered multiuser systems.
\begin{Theorem}
For multiuser uplink systems, \LASC becomes
\begin{equation}
  \begin{cases}
    \epsilon^{*}\left[k\right] =  F_{\eta}\left[ \frac{1}{M^{1 - \rho\left[k\right]}} \left(1 + \frac{K}{\tau\left[k\right] } +  p_{I}\right) \frac{1}{\gamma P}\right] \\
    \mu^{*}\left[k\right]: r_{t}\left[k\right] = \min\left(q_{t}\left[k\right], \ 2 \lambda\left[k\right]  \right).
  \end{cases}~\label{equ:eps_l_mu_l_def_MU}
\end{equation}
As $M\to \infty$, for positive $\tau\left[k\right]$ and $\rho\left[k\right]\in\left[0,1\right)$, each user operates under \LASC achieves the minimum latency of
\begin{equation}
  D^{*}\left[k\right] - 1\cong \epsilon^{*} \left[k\right], \ k= 1, 2, \dots, K, \ M \to \infty.
\end{equation}
\label{theorem:mu_cha_admissive}
\end{Theorem}
\begin{proof}
With random matrix theory, we prove by adopting similar steps as in the single-user case. The key is step is to compute the mean and variance of~\eqref{equ:kappa_mu}.
Please find the proof in Appendix~\ref{appd:wishart}.
\end{proof}

Recall that $f_{1}\left(M\right) \cong f_2\left(M\right)$ denotes that $\lim_{M \to \infty} \frac{f_1\left(M\right)}{f_2\left(M\right)} = 1$.
\LASCNoSpace, therefore, indeed provides the latency-optimal target error rate and transmission rate policies to the multiuser massive MIMO system. And Theorem~\ref{theorem:mu_cha_admissive} also captures the minimum latency of each user.

In conclusion, for any non-negative weights $\omega_{k}$, we can convert the $K$ user optimization problem into $K$ parallel single user problems.
For finite $M$, Algorithm~\ref{alg:dp} solves each of the single user problems and provides the optimal target error rate and transmission rate policy.
Furthermore, each user operates using \LASC distributedly is asymptotically latency-optimal.

We end this section by discussing some possible extensions of the multiuser system analysis.

The first extension is the general multiuser MIMO systems with user correlation.
For massive MIMO, the user channels are expected to become mutually orthogonal as $M$ increases, which is usually referred to as ``favorable propagation''~\cite{marzetta2016fundamentals,bjornson2017massive}.
The favorable proportion is expected to hold in massive MIMO systems~\cite{marzetta2016fundamentals,bjornson2017massive} and is verified by recent massive MIMO measurements~\cite{7062910, zero_force_closed_form}.
However, for small scale multiuser systems, user channels might be significantly correlated, and the multiuser scheduling problem cannot be fully decoupled.
While spatial multiplexing correlated user leads to smaller $\mathrm{SINR}$, spatial multiplexing only non-correlated users can lead to longer queueing latency.
Hence, we expect a latency-minimizing scheduler should balance a tradeoff between longer queueing time and smaller $\mathrm{SINR}$.

The second extension is to model the pilot contamination and base-station array correlation, which both can reduce the $\sinr$.
The pilot contamination~\cite{marzetta2016fundamentals, bjornson2017massive} is caused by pilot reuse and leads to both non-coherent and coherent interference.
In particular, without proper pilot decontamination, coherent interference can grow linearly with the number of base-station antennas.
Recent research~\cite{8094949,bjornson2017massive} demonstrates that via multicell joint
transmission, the massive MIMO system can reject the coherent interference if the covariance matrix of pilot sharing users is asymptotically linearly independent.
Under the same condition,~\cite{8094949,bjornson2017massive} shows that the effective $\sinr$ can grow linearly with $M$ without bound with pilot contamination and base-station array correlation.
Therefore, it is reasonable to use a finite $p_{I}$ to model the power of the residual inter-cell interference after pilot decontamination.

Finally, we consider the latency-minimum transmission control of multicell systems with pilot contamination and base-station array correlation as an important future work.
Note that~\cite{8094949, bjornson2017massive} shows that the $\sinr$ can also grow linearly with $M$, which implies that the mean of the per-antenna gain would be lower bounded by a positive constant.
Computing the variance condition and finding the optimal transmission control for this generalized setup is beyond the scope of this paper.
To evaluate the impact of the spatial correlation, we utilize over-the-air measured channels in Section~\ref{sec:numerical}.

\section{Numerical Results}~\label{sec:numerical}
In this section, we utilize measured channels and simulated channels to confirm our previous analysis in Section~\ref{sec:algo} and Section~\ref{sec:MU}.
During the numerical evaluation, the latency duration is captured in the unit of second, which is obtained by multiplying frame duration to latency measured in the unit of frame.
We measure the over-the-air channels between mobile clients and a $64$-antenna massive MIMO base-station with Argos system~\cite{shepard2012argos} on the campus of Rice University. Figure~\ref{fig:array} and~\ref{fig:exp_setup} describes the Argos array and the over-the-air measurement setup.
We measured the $2.4$ GHz Wi-Fi channel ($20$ MHz, $52$ non-empty data subcarriers) for four pedestrian users in non-line-of-sight environments, which are denoted by Fig.~\ref{fig:exp-loc}. For each user, we take channel measurements over $7900$ frames of all subcarriers. The effective measured $\snr$ between each mobile user and each base-station antenna is higher than $15$ dB.
In simulations, we consider measured over-the-air channel traces as the perfect channel.

\begin{figure*}[htbp]
\centering
  \subfloat[Argos Array]{ \label{fig:array}
\includegraphics[height=0.18\textwidth]{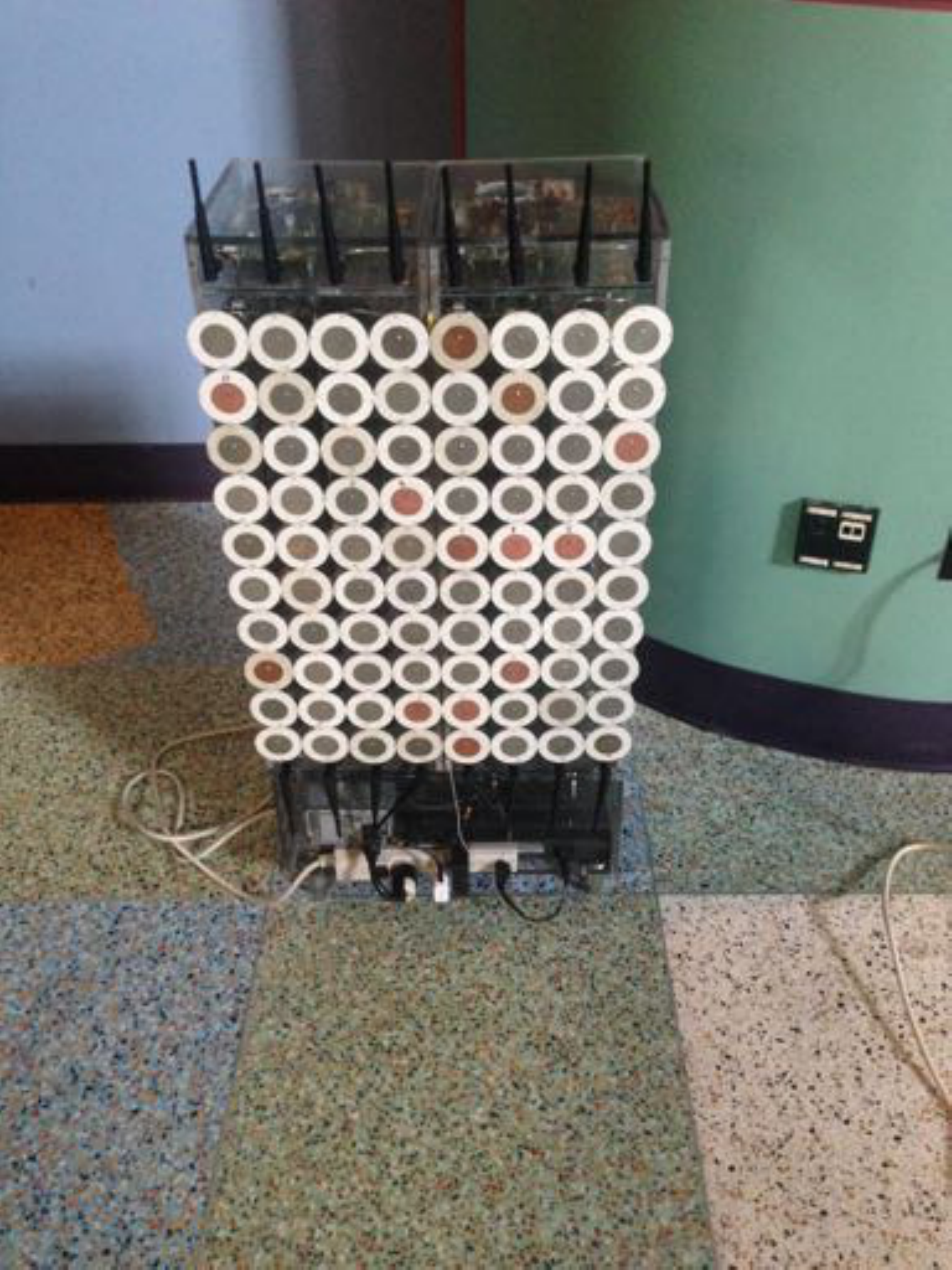}
  }
  \quad
  \subfloat[Over-the-air Measurement Setup]{ \label{fig:exp_setup}
\includegraphics[height=0.18\textwidth]{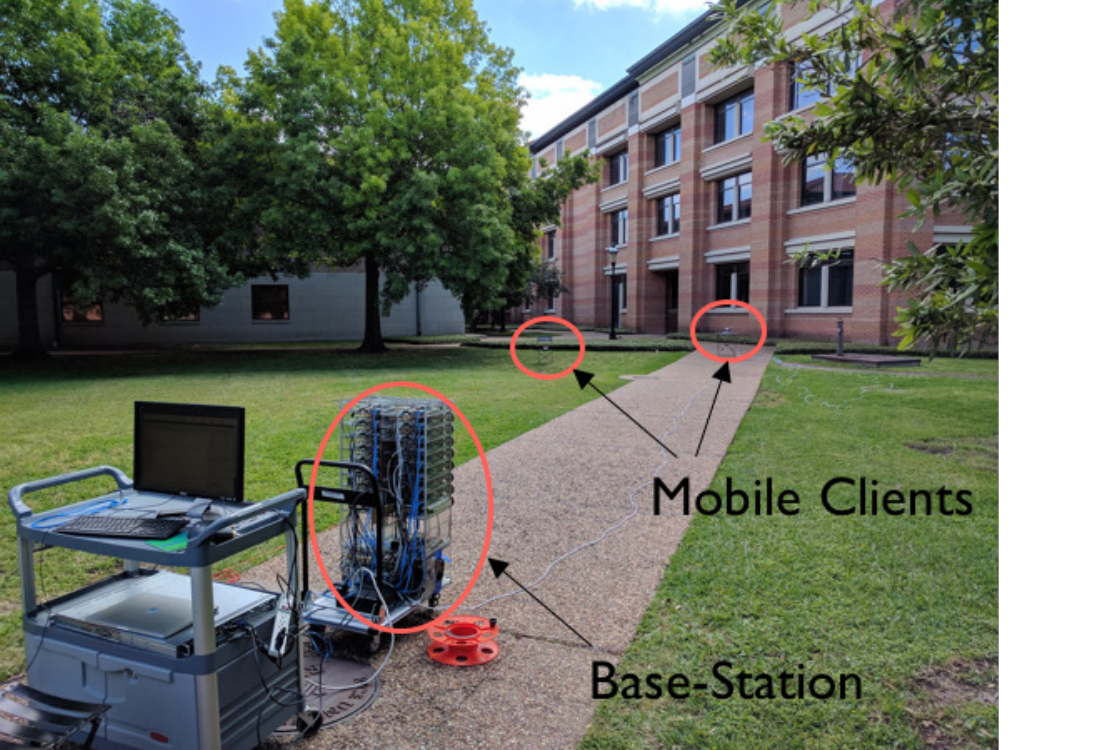}
  }
  \subfloat[User Locations]{ \label{fig:exp-loc}
  \includegraphics[height=0.18\textwidth]{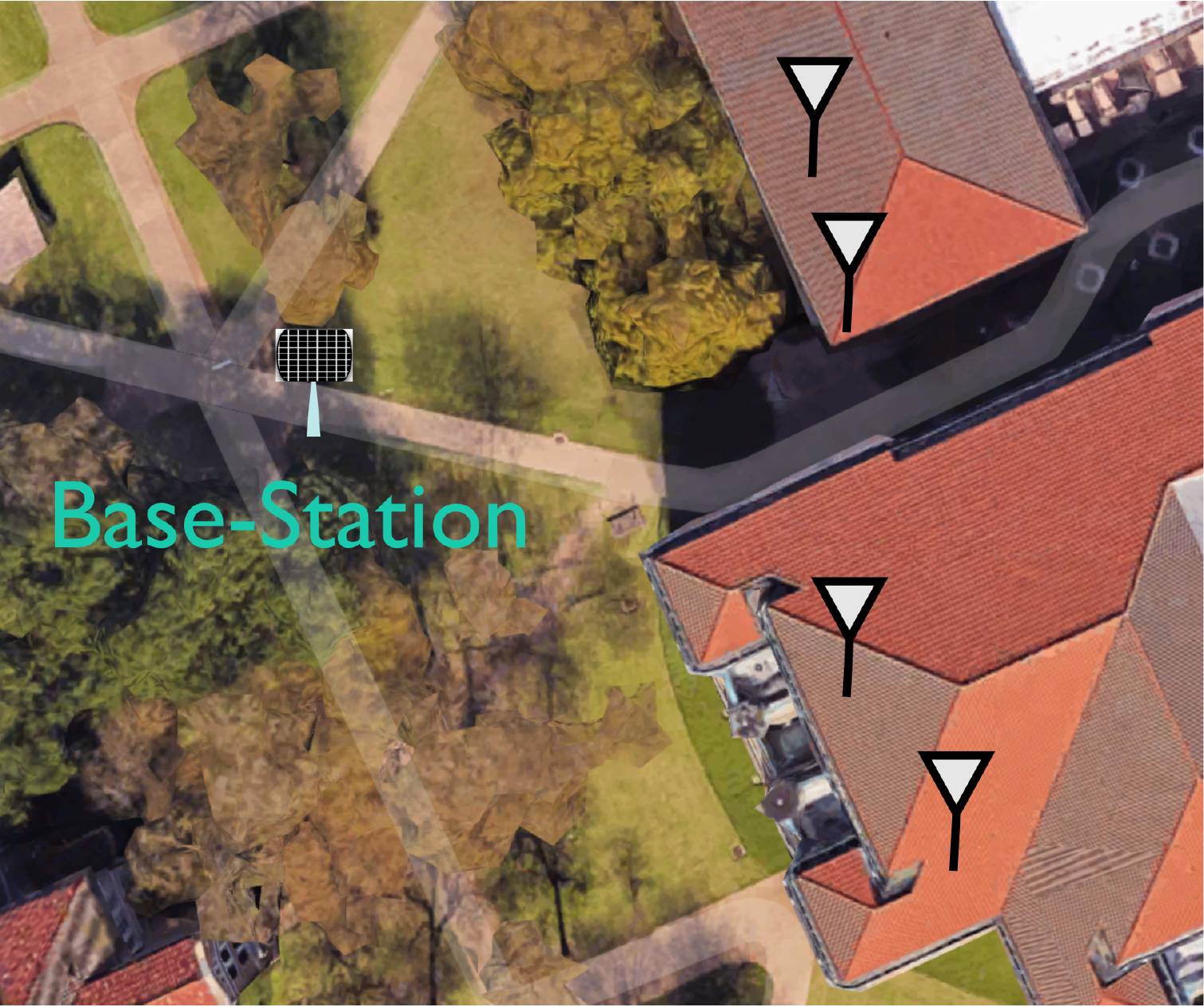}
    }
\caption{Argos~\cite{shepard2012argos} Massive MIMO base-station and the over-the-air measurements setup. The background map of Fig.~\ref{fig:exp-loc} is generated by Google Maps~\cite{google_map}.
The black single antennas denotes the locations of the mobile users.}
\end{figure*}

The base-station adopts MMSE estimator to estimate $\tau$ uplink pilots, each of power $20$ dBm, from the users.
Using the estimated channel, the base-station generates zero-forcing receive beamformers to decode the signal of each user.
The users are assumed to follow average power constraint of $20$ dBm with large-scale fading of $-10$ dB. The maximum buffer length $B$ is $10$.
The packet arrival rate is uniform over the time at the rate of $5$ packets per frame.
And the packet size $L$ is $52$ bits per OFDM symbol.
The latency penalty of dropped packets from buffer overflow is $0.5$ s.
And each self-contained frame is considered of duration $0.25$ ms.
The state space of the target error rate is $[1\%, \ 2\%, \ \dots, \ 20\%]$, $[0.1\%, \ 0.2\%, \ \dots, \ 0.9\%]$, and $[0.01\%, \ 0.02\%, \ \dots, \ 0.09\%]$.
Each user is under a maximum target error rate constraint of $3.16$\%, which is equivalent to the 5G URLLC reliability constraint of $99.9999$\% (over $1$ ms).
And the power of the inter-cell interference equals the receiver noise floor.

\begin{figure*}[htbp]
    \centering
  \subfloat[$\tau = 2$, Measured Channel\label{1a}]{%
       \includegraphics[width=0.42\linewidth]{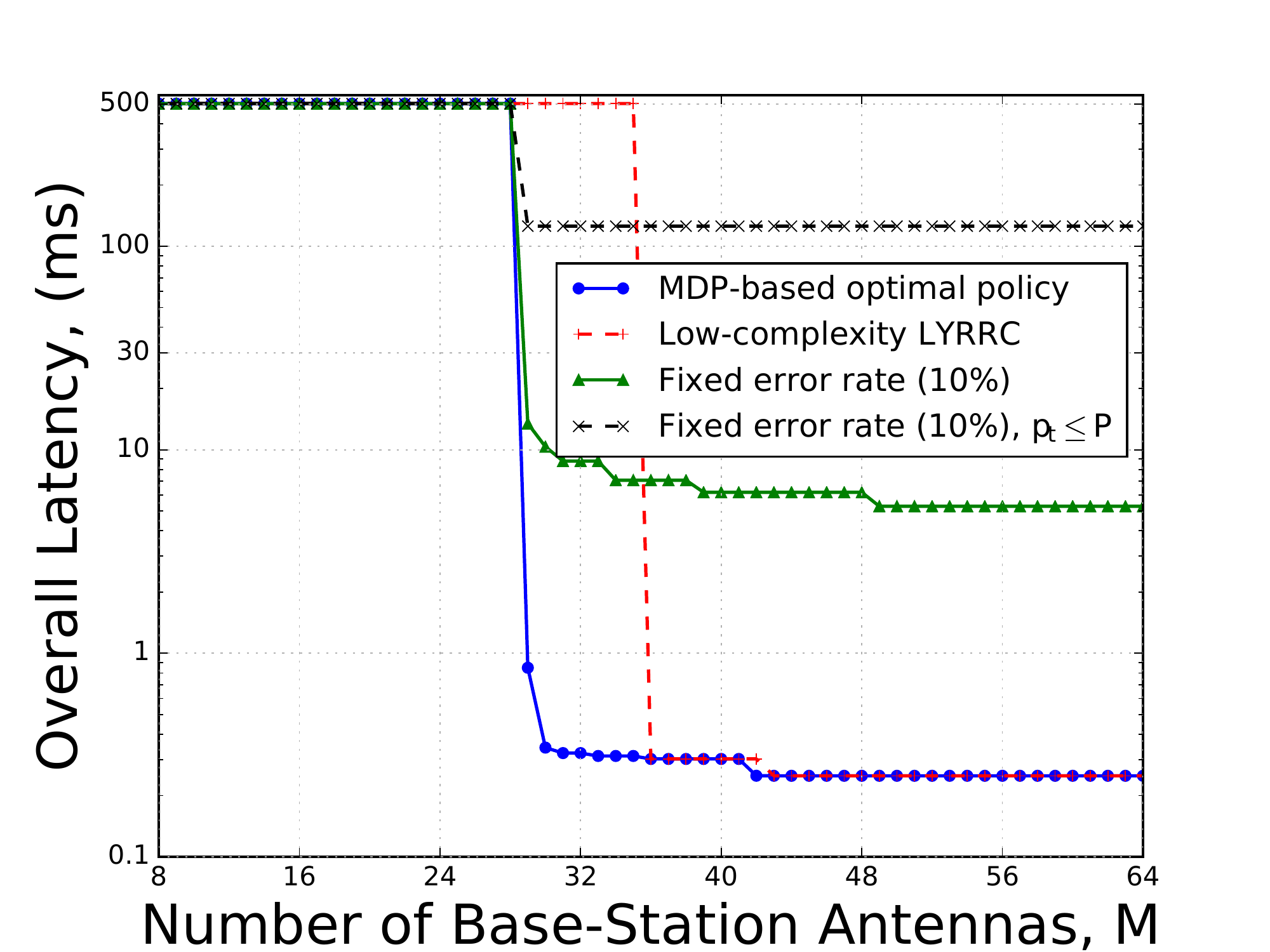}}
    \quad
  \subfloat[$\tau = 4$, Measured Channel\label{1b}]{%
        \includegraphics[width=0.42\linewidth]{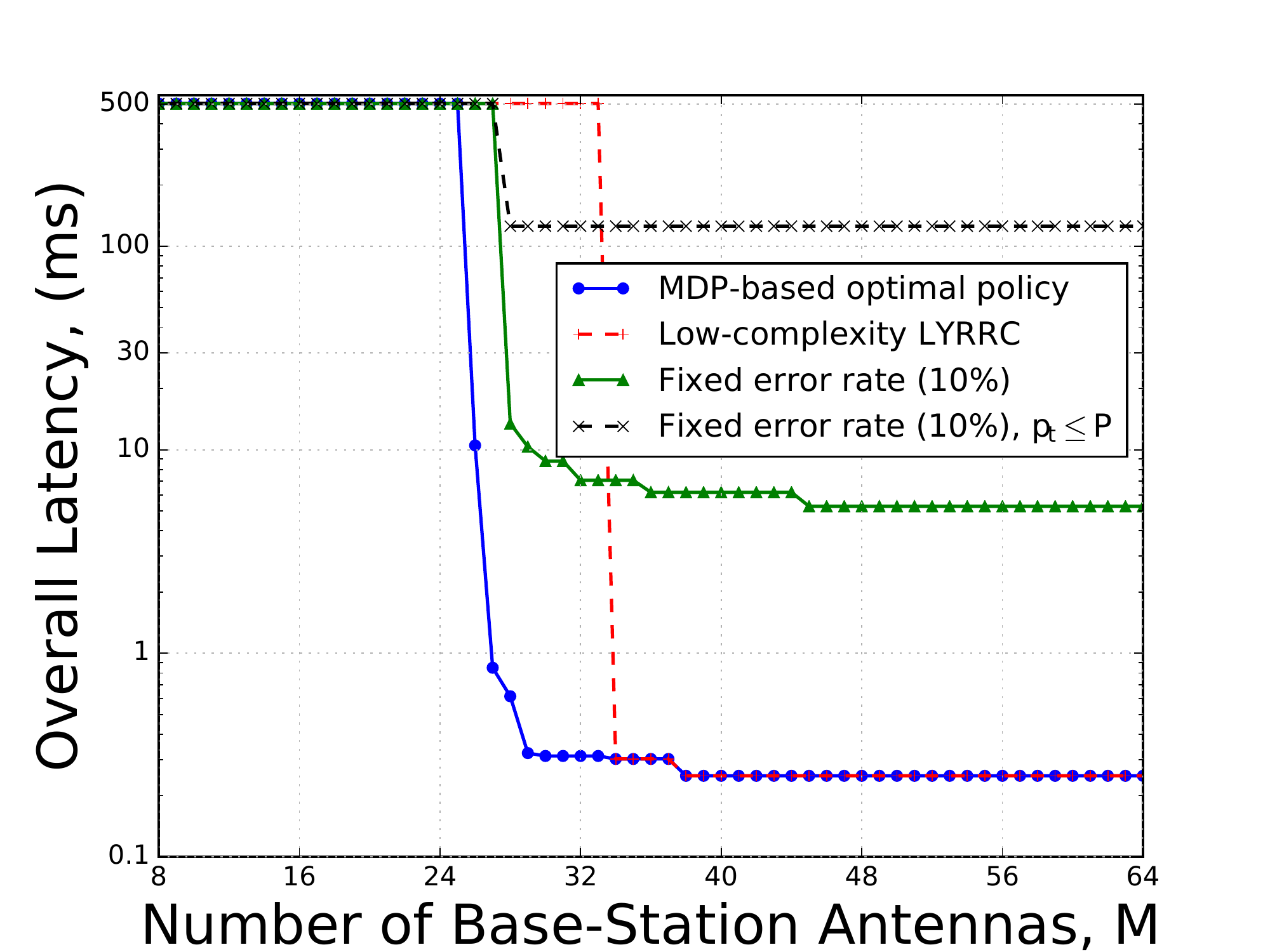}}
    \\
  \subfloat[$\tau = 2$, i.i.d. Rayleigh Fading\label{1c}]{%
        \includegraphics[width=0.42\linewidth]{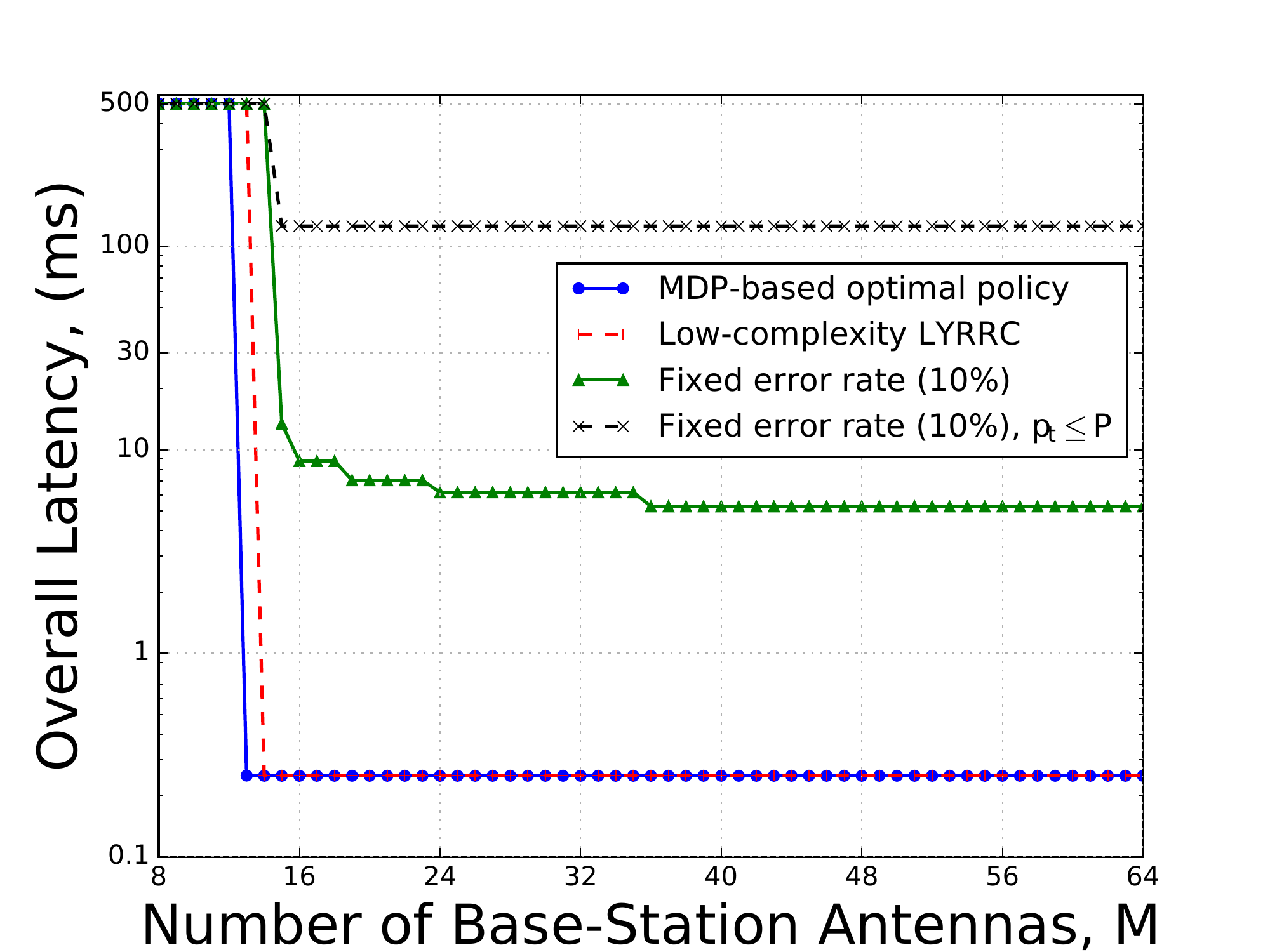}}
    \quad
  \subfloat[$\tau = 4$, i.i.d. Rayleigh Fading\label{1d}]{%
        \includegraphics[width=0.42\linewidth]{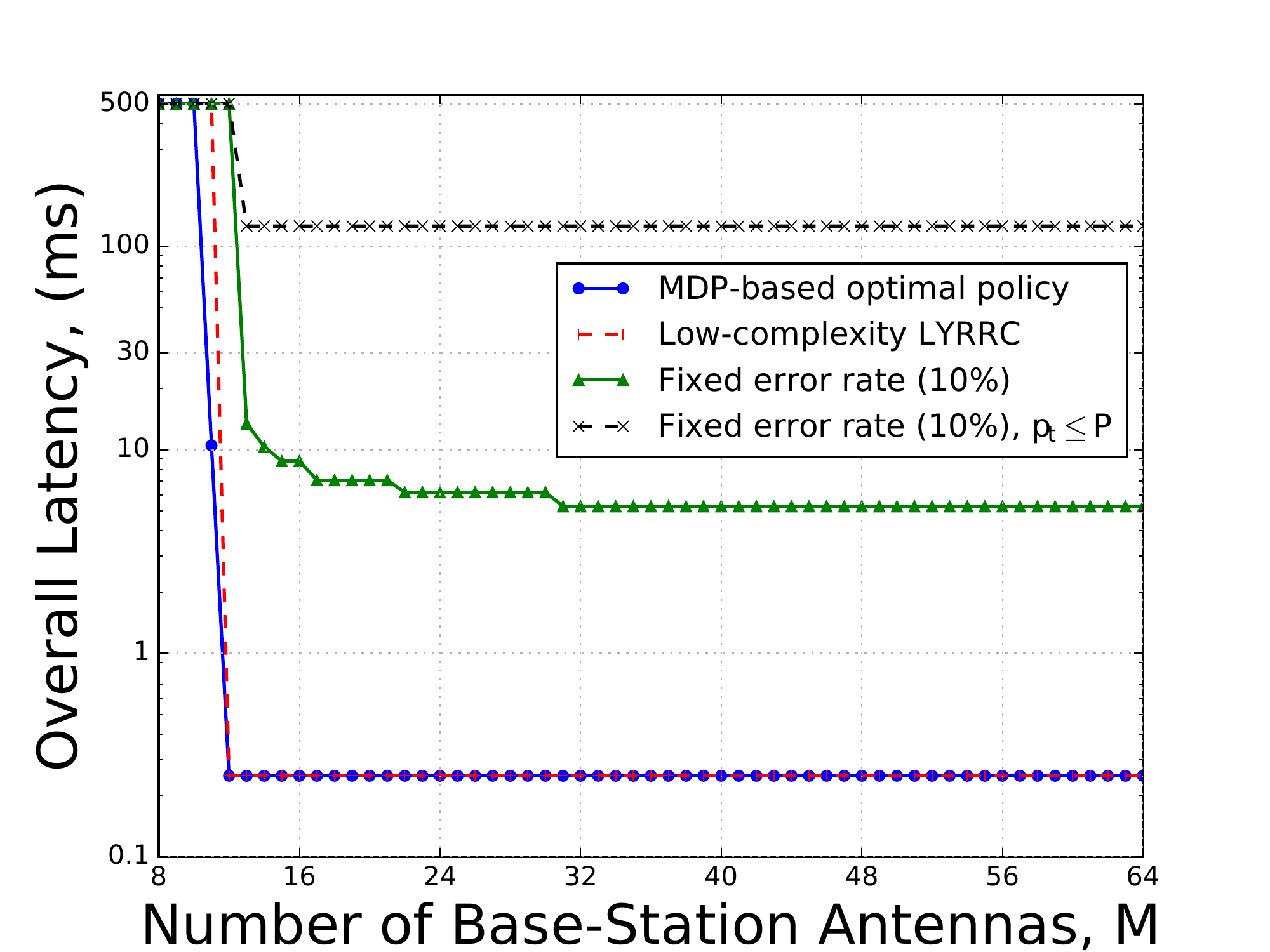}}
\caption{
The solved latency under four different policies over measured and simulated $4$-user channels. Algorithm~\ref{alg:dp} generated policy and \LASC\eqref{equ:eps_l_mu_l_def} are labeled by blue and red,respectively. Green lines is the policy of fixed reliability (target error rate) and power adaptation based on queue-length. The peak power constrained policy with fixed reliability is in black. The traffic arrival rate is $5$ packets per frame, each of size $52$-bits per OFDM symbol. The pilots power is $20$ dBm. The average power constraint is $20$ dBm with large-scale fading of $-10$ dBm.}~\label{fig:delay}
\end{figure*}

Fig.~\ref{fig:delay} provides the latency performance comparison of four different policies over the measured channels and simulated i.i.d. Rayleigh fading channels.
The blue lines are the optimal array-latency curves under the proposed joint reliability and transmission rate adaptation, which is obtained by Algorithm~\ref{alg:dp}.
The red lines are the proposed low-complexity \LASC\eqref{equ:eps_l_mu_l_def}, which was discussed in Section~\ref{sec:large_M}.
The green colored lines capture the latency under optimal transmission rate adaptation but fixed reliability (target error rate of $10\%$).
And the black lines are the latencies of fixed reliability ($10\%$ target error rate) and transmission rate adaptation under a peak power constraint, which is currently deployed in LTE and Wi-Fi systems.

Over measured and simulated channels, the proposed joint control (blue and red lines) clearly provides better latency performance than the two fixed-reliability counterparts.
Allowing target error rate to be adaptive on the number of antennas $M$ turned out to reduce the latency significantly.
Compared to the fixed target error rate with peak power control, a $400 \times$ latency reduction is observed when $M>30$.
Additionally, when $M$ is larger than $30$, we find that the proposed joint control can provide a $20 \times$ latency reduction compared to the state-of-the-art control that fixes target error rate and adapts transmission rate~\cite{berry2002communication, 4294166, 4567575, cao2008power} (based on the
number of antennas and queue length).
The large-array asymptotic latency-optimal control, \LASCNoSpace, turned out to be near latency-optimal when $M$ is larger than $30$.
Finally, we find policies that fixed target error rate at $10$\% leads to at least $5$ ms latency and cannot satisfy the URLLC latency requirement.

Fig.~\ref{fig:delay} captures the influence of imperfect channel state information on latency.
For a multiuser uplink system, the inter-beam interference~\eqref{equ:kappa_mu} reduces with the number of pilots $\tau$.
And achieving the same target error rate becomes more power expensive with larger inter-beam interference.
Therefore, over measured and simulated channels, the latency increases as $\tau$ reduces.

Fig.~\ref{fig:delay} also demonstrates that the spatial correlation of the base-station antennas reduces the minimum achievable latency.
With the same number of pilots $\tau$, a lower latency is observed in i.i.d. Rayleigh fading channels than that in measured channels.
The increased latency can be explained by the reduced system capacity from spatial correlation~\cite{marzetta2016fundamentals,bjornson2017massive}.
We further remark that \LASC achieves near optimal latency performance over both measured and simulated channels when $M>36$.

\begin{figure*}[htbp]
\centering
\subfloat[Latency-optimal Target Error Rate and Number of Base-Station Antennas $M$]{\label{fig:eps_vs_M}
\includegraphics[width=0.4\textwidth]{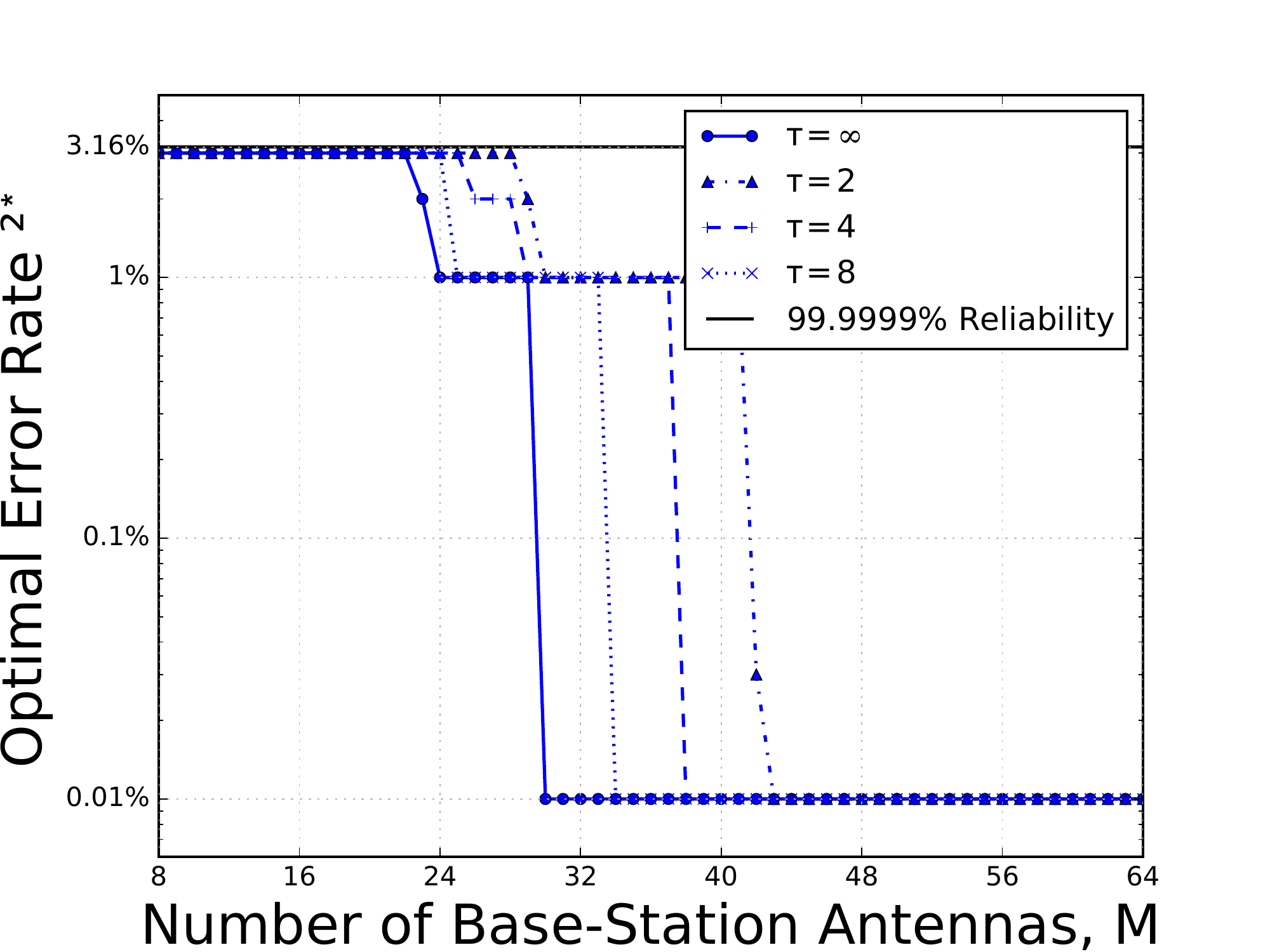}
}
\subfloat[Latency under $\mu^{l}$~\eqref{equ:eps_l_mu_l_def} with finite buffer length.]{\label{fig:mu_l_delay}
\includegraphics[width=0.4\textwidth]{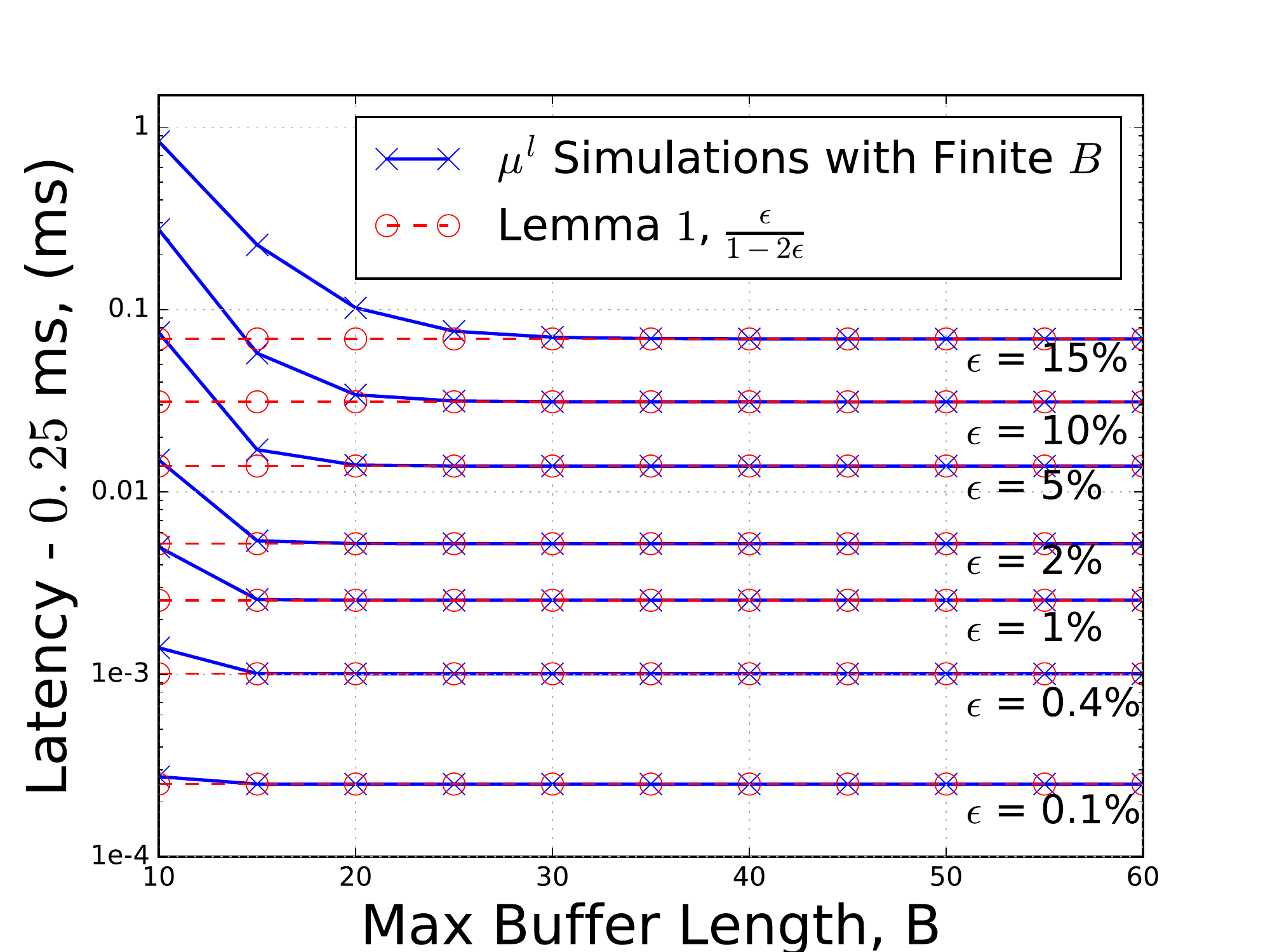}
}
\caption{Fig.~\ref{fig:eps_vs_M} shows the computed error rate that provides minimum latency in the measured channels.
And the resulted minimum latencies are shown in Fig.~\ref{fig:delay} (in blue).
Fig.~\ref{fig:mu_l_delay} verifies the latency characterization under ``rule of double'' in Lemma~\ref{lemma:tx_rate}.
}
\end{figure*}

We now comment on the optimal target error rate that minimizes the latency.
Fig.~\ref{fig:eps_vs_M} describes the latency-optimal target error rate obtained during solving the latency minimization problems in Fig.~\ref{fig:delay}.
The latency-optimal target error rate increases as $\tau$ reduces due to less accurate channel estimation, which agrees with \LASCNoSpace.
Additionally, due to the reliability constraint, the solved latency-optimal target error rates satisfy the 5G reliability requirement (target error rate of $3.16\%$).

Finally, we use simulations to verify our structural analysis in Section~\ref{sec:large_M}.
Fig.~\ref{fig:delay} confirms that \LASC\eqref{equ:eps_l_mu_l_def} is near latency-optimal for $M$ larger than a finite number of $38$.
One technical contribution independent of the massive MIMO system is a simple transmission rate adaptation $\mu^{l}$ as $\min\left(q, \ 2 \lambda  \right)$, which is referred to as ``rule of double'' and is part of~\LASCNoSpace.
Lemma~\ref{lemma:tx_rate} captures that, when buffer size $B\to \infty$, the resulted latency by using $\mu^{l}$ and a target error rate $\epsilon<0.5$ is $1 + \frac{\epsilon}{1-2\epsilon}$.
Fig.~\ref{fig:mu_l_delay} shows the resulted latency by using $\mu^{l}$ with a finite buffer size.
The (large-buffer) asymptotic latency turned out to accurately approximate the system latency when $B$ is larger than $30$.
And as the target reliability increases (target error rate reduces), buffer overflow is less likely to happen and the latency approximation in Lemma~\ref{lemma:tx_rate} becomes increasingly accurate.

\section{Conclusion}\label{sec:Conclude}
In this work, we study the latency-optimal cross-layer control over wideband massive MIMO channels.
By identifying a tradeoff between queueing and retransmission latency, we find that a lower physical layer target error rate does not always guarantee lower latency.
We present algorithms that generate the optimal target error rate and transmission rate policies.
We show that to achieve the minimum latency, the target error rate can no longer be considered fixed and needs to be adapted based on the number of base-station antennas, channel estimation accuracy, and the traffic arrival rate.
Our results also demonstrate that massive MIMO systems have the potential to achieve both high reliability and low latency and are a promising candidates of 5G URLLC.

\begin{appendices}
\section{Proof of Theorem~\ref{thm:large_array_lower_bd}}\label{appendix:delay_gap}
  We use a per packet argument. Since infinite buffer is assumed in this section, no packet is dropped and {\sl all} packets will be successfully received with a variable number of transmissions due to the potential channel-induced error. For any target error rate $\epsilon$, let $r$ be the average number of retransmissions. The sum of the retransmission latency and transmission time equals
  \begin{equation}
    1 + \sum_{r=0}^{\infty} \Prob\left(r\right) r = 1 + \sum_{r=0}^{\infty} r \left(1-\epsilon\right)\epsilon^{r} =1 + \frac{\epsilon}{1-\epsilon}, ~\label{equ:large_array_delay_gap_lower_bd}
  \end{equation}
  which is a lower bound of the total latency because the queueing latency is ignored.
  To finish the proof, we now lower bound $\epsilon$ under the long-term power constraint $P$.
  Under the steady state, the average transmission rate equals to the packet arrival rate, i.e.,
  \begin{equation}
   \lambda = E_{\pi}\left[ r \left(1 - \epsilon\right) \right] = E_{\pi}\left[r\right]\left(1 - \epsilon\right). ~\label{equ:d_lbd_A}
  \end{equation}
  The power function~\eqref{equ:p_per_frm} is convex on $r$. We can apply Jensen's inequality and~\eqref{equ:rho_def} to obtain a lower bound for the average transmission power as
\begin{align}
P &= E_{\pi}\left[ p\left(r, \epsilon, \gamma \right) \right] \notag\\&\geq
\left\{\frac{\gamma F_{\eta}^{-1}\left(\epsilon\right)}{
\exp\left[\left(\frac{\rho}{1-\epsilon}-1\right) \log M\right]
} - \frac{\gamma}{1 + \gamma p_{\tau} \tau } \right\}^{-1},~\label{equ:P_lower_appd}
  \end{align}
Function $F_{\eta}^{-1}$ is an inverse CDF and is non-decreasing.
From~\eqref{equ:P_lower_appd}, the $\epsilon$ is lower bounded as
\begin{equation*}
F_{\eta}^{-1}\left(\epsilon\right) \geq
M^{-\left[1-\rho/\left(1-\epsilon\right)\right]}
  \left(\frac{1}{\gamma P} + \frac{1}{\gamma p_{\tau} \tau}\right).
  \end{equation*}
Using the monotonicity of the CDF, a lower bound on the target error rate $\epsilon$ is then
\begin{equation}
  \epsilon \geq F_{\eta}\left[ \frac{1}{M^{\left(1 - \rho\right)}} \left(\frac{1}{\gamma P} + \frac{1}{\gamma p_{\tau} \tau}\right)\right].
  ~\label{equ:d_lbd_eps_lbd}
  \end{equation}
We finish the proof by combining~\eqref{equ:d_lbd_eps_lbd} and~\eqref{equ:large_array_delay_gap_lower_bd}.

\section{Proof of Lemma~\ref{lemma:tx_rate}}\label{appendix:tx_rate}
We compute the queueing latency by considering the steady state. Under transmission rate adaptation $\mu^{l}$, the buffer length process~\eqref{equ:buff_evlo} is rewritten as $q_{t+1} = \max\left[q_{t} +\left(1-2 \ 1_{t}\right)\lambda, \lambda\right].$
The buffer length process under $\mu^{l}$ thus constitutes a Markov chain with {\sl countably infinite states}~\cite{gallager2012discrete}.
The distribution of $1_{t}$ is determined by target error rate $\epsilon$ as $\Prob\left(1_{t}=1\right)=\epsilon$ and $\Prob\left(1_{t}=0\right)=1-\epsilon$.  The state transition is shown in Fig.~\ref{fig:q_mu_l}. Denote the steady state distribution of the buffer length as $\pi_{q}$. We then have that
\begin{equation}
\begin{cases}
\pi_{\lambda} &= \left(1-\epsilon\right) \pi_{\lambda} + \left(1-\epsilon\right) \pi_{2\lambda} \notag \\
\pi_{i\lambda} &= \epsilon \pi_{\left(i-1\right)\lambda}  + \left(1-\epsilon\right) \pi_{\left(i+1\right)\lambda} , \quad i \geq 2,
\end{cases} \notag
\end{equation}
where $\sum_{i=0}^{N} \pi_{i\lambda} = 1$. The steady state distribution is then computed as
\begin{equation}
\pi_{i\lambda} = \left(1-\frac{\epsilon}{1-\epsilon}\right) \left(\frac{\epsilon}{1-\epsilon}\right)^{i-1},\quad  i=1,2,\dots.~\label{equ:pi_q_mu_l}
\end{equation}
Using~\eqref{equ:pi_q_mu_l}, the average latency is then computed as
\begin{align*}
\frac{1}{\lambda}E_{\pi_{q}}\left[q\right]& = \frac{1}{\lambda}\left(\sum_{i=1}^{\infty} \pi_{i\lambda} i\lambda\right)
\\&=  \sum_{i=1}^{\infty} \left(\frac{\epsilon}{1-\epsilon}\right)^{i-1} i -  \sum_{i=1}^{\infty} \left(\frac{\epsilon}{1-\epsilon}\right)^{i} i
= 1 + \frac{\epsilon}{1-2\epsilon},
\end{align*}
which completes the proof.


\section{Proof of Theorem~\ref{thm:large_array_contrl}}\label{appendix:large_array_contrl}
We characterize the gap between latency under LYRRC as
\begin{align}
D_{\mathrm{LYRRC}} - D^{*}\left(M\right) &= \left(D_{\mathrm{LYRRC}} - 1\right) - \left(D^{*} -1\right)
\notag \\
&\leq \frac{\epsilon_{o}}{1-2\epsilon_{o}} - \frac{\epsilon_{o}}{1-\epsilon_{o}} = \frac{\left(\epsilon_{o}\right)^2}{\left(1 - 2\epsilon_{o}\right)\left(1-\epsilon_{o}\right)},~\label{equ:tmp3}
\end{align}
where the last step is obtained via applying Theorem~\ref{thm:large_array_lower_bd} and~\eqref{equ:d_lbd_eps_lbd}. Equ.~\eqref{equ:tmp3} provides the characterization of the latency gap.
To finish the proof, it is sufficient to show that the average power constraint $P$ is satisfied under the large-array simple control.

With utilization factor $\rho$~\eqref{equ:rho_def}, the packet arrival rate scales as $\lambda=\left(\rho N \log M\right)/L$.
Using the per-frame power~\eqref{equ:p_per_frm} and the definition of $\epsilon_{o}$~\eqref{equ:eps_l_mu_l_def}, the transmission power with rate $r$ is
\begin{equation}
  P^{\epsilon_{o}}\left(r\right)
  =  \left[
\left(\frac{1}{P} + \frac{\gamma}{\tau \gamma p_{\tau} + 1}\right) \frac{1}{M^{\rho\left(r/\lambda -1\right)}} - \frac{\gamma}{\tau \gamma p_{\tau} + 1} \right]^{-1}. ~\label{equ:p_eps_l}
\end{equation}
Since we assume empty buffer at time $0$ and constant arrival rate of $\lambda$, the transmission rates under policy $\mu^{l}$ is either $\lambda$ or $2\lambda$. Based on the queue length steady state characterization~\eqref{equ:pi_q_mu_l}, we have that $\Prob(u=\lambda) =\pi_{\lambda}= 1- \frac{\epsilon_{o}}{1-\epsilon_{o}}$
and $\Prob(r=2\lambda) =\sum_{i=2}^{\infty}\pi_{i\lambda}= \frac{\epsilon_{o}}{1-\epsilon_{o}}$. Conditioning on the rate expression in~\eqref{equ:p_eps_l}, the average power under \LASC is
\begin{align}
     P^{\epsilon_{o},\mu^{l}} &=
     \frac{1-2\epsilon_{o}}{1-\epsilon_{o}}P^{\epsilon_{o}}\left(\lambda\right) + \frac{\epsilon_{o}}{1-\epsilon_{o}} P^{\epsilon_{o}}\left(2\lambda\right)\notag \\&= \frac{1-2\epsilon_{o}}{1-\epsilon_{o}}P + \frac{\epsilon_{o}}{1-\epsilon_{o}} P^{\epsilon_{o}}\left(2\lambda\right).\label{equ:P_mu_l_eps_L}
\end{align}
We want to show that the power constraint is satisfied, i.e., $P^{\epsilon_{o},\mu^{l}}\leq P$.
Using~\eqref{equ:p_eps_l}, the second power consumption term of~\eqref{equ:P_mu_l_eps_L} is upper bounded as
\begin{equation}
\frac{\epsilon_{o}}{1-\epsilon_{o}}  P^{\epsilon_{o}}\left(2\lambda\right) \leq \epsilon_{o}P^{\epsilon_{o}}\left(2\lambda\right)  \leq \epsilon_{o} M^{\rho}.
\end{equation}
Therefore, the sufficient condition~\eqref{equ:P_mu_l_eps_L} is equivalent to
\begin{equation}
\lim_{M\to\infty}\epsilon_{o}\exp\left(\rho \log M\right) = \lim_{M\to\infty}\epsilon_{o}M^{\rho}= 0.~\label{equ:sufficient}
\end{equation}
Before proving~\eqref{equ:sufficient}, we first present an upper bound of $\epsilon_{o}$.
The effective channel gain $\eta$~\eqref{equ:eta_def} is the average of $N$ i.i.d. random variables $\log \kappa$. For $x<0$, we thus have an upper bound as
\begin{align}
F_{\eta}\left(x\right)& = F_{\sum_{n=1}^{N}\log\kappa_{n}}\left(Nx\right) \leq F_{\log\kappa}\left(Nx\right)
\notag\\&= Pr\left(\kappa \leq \exp\left(Nx\right)\right),~\label{equ:F_up_bd_1}
\end{align}
where the last step is by the definition of CDF. We now upper-bound~\eqref{equ:F_up_bd_1} as the follows.
\begin{align*}
F_{\eta}\left(x\right) &\leq Pr\left(\kappa - E\left[\kappa\right] \leq  \exp\left(Nx\right) - E\left[\kappa\right]\right)
\\
&\leq Pr\left(|\kappa - E\left[\kappa\right] | \geq  E\left[\kappa\right] - \exp\left(Nx\right)\right).
\end{align*}
Here, the last term denotes the probability that $\kappa$ has a larger deviation (to its mean) than
$E\left[\kappa\right] - \exp\left(Nx\right)$. Using Chebyshev's Inequality, a new upper-bound is obtained as
\begin{align}
F_{\eta}\left(x\right) &\leq \frac{\Var\left[\kappa\right]}{(E\left[\kappa\right] - \exp\left(Nx\right))^2}
\notag\\&=  \frac{1}{M} \frac{1}{(\frac{\tau p_{\tau} \gamma}{1 + \tau p_{\tau} \gamma} - \exp\left(Nx\right))^2} \left(\frac{\tau p_{\tau} \gamma}{1 + \tau p_{\tau} \gamma}\right)^2 = O\left(\frac{1}{M}\right),  \label{equ:final_cond}
\end{align}
where the last step is by conditions~\eqref{equ:admissive_mean} and~\eqref{equ:admissive_variance}.
By the definition of $\epsilon_{o}$, using the above upper bound proves~\eqref{equ:sufficient} and completes the proof.

\section{Proof of Theorem~\ref{theorem:mu_cha_admissive}}~\label{appd:wishart}
The multi-user mapping~\eqref{equ:MU_power_Map} can be viewed as a scaled version of~\eqref{equ:p_per_frm} when $\tau=\infty$.
Recall that the proof of Theorem~\ref{thm:large_array_lower_bd} and Lemma~\ref{lemma:tx_rate} is independent of the distribution of the per-antenna gain $\kappa_{n}$. To complete the proof, we only need to prove the mulituser version of Theorem~\ref{thm:large_array_contrl} by following the same derivations as in Appendix~\ref{appendix:large_array_contrl}.
As the proof from~\eqref{equ:tmp3} to~\eqref{equ:F_up_bd_1} is also independent to the distribution of $\kappa_{n}$, we finish the proof by proving that~\eqref{equ:kappa_mu} satisfies~\eqref{equ:final_cond}.
By the multiuser setup in Section~\ref{sec:MU}, $\kappa_{n}$~\eqref{equ:kappa_mu} equals
$
\frac{\tau p_{\tau}\left[k\right]\gamma\left[k\right]}{\tau p_{\tau}\left[k\right]\gamma\left[k\right] + 1} \frac{1}{M \left[\mathbf{W}^{-1}\right]_{kk}},
$
where $\mathbf{W}$ is a $K \times K$ central complex Wishart matrix with $M$ degrees of freedom and covariance matrix of $\mathbf{I}$. Since $\frac{\tau p_{\tau}\left[k\right]\gamma\left[k\right]}{\tau p_{\tau}\left[k\right]\gamma\left[k\right] + 1}$ is a positive constant, we only need to capture the mean and variance of  $\frac{1}{M \left[\mathbf{W}^{-1}\right]_{kk}}$to verify~\eqref{equ:final_cond}.
We first check the mean condition by Jensen's inequality as
$
  E\left[\frac{1}{M \left[\mathbf{W}^{-1}\right]_{kk}}\right] \geq
  \frac{1}{E\left[M \left[\mathbf{W}^{-1}\right]_{kk}\right]}\notag
$.
Using the first moments of inverse Wishart~\cite{tulino2004random} gives that
\begin{equation}
E\left[M \left[\mathbf{W}^{-1}\right]_{kk}\right] =
\frac{1}{K}E\left[M \tr\left(\mathbf{W}^{-1}\right)\right] = \frac{M}{M-K}.
\end{equation}
Therefore, the per-antenna gain is lower bounded by a constant as $M\to \infty$.
Recall that the $\kappa_{n}$ in systems with perfect channel case serves as an upper bound. In the upper bound case, the per-antenna gain expectation is $1$ as $M\to\infty$.
By random matrix theory~\cite{tulino2004random}, the variance of the trace of inverse Wishart satisfies
\begin{align}
  \Var\left[\tr\left(\mathbf{W}^{-1}\right)\right]&=E\left\{\left[\tr\left(\mathbf{W}^{-1}\right)\right]^2\right\} - E\left[\tr\left(\mathbf{W}^{-1}\right)\right]^2 \notag
  \\&=\frac{MK}{\left(\left(M-K\right)^2-1\right)\left(M-K\right)^2}.\notag
\end{align}
Using Taylor's expansion, we complete the proof by checking the variance as
\begin{align}
  \Var\left[\frac{1}{M\left[\mathbf{W}^{-1}\right]_{kk}}\right] &= \frac{\Var\left[M\tr\left(\mathbf{W}^{-1}\right)\right]_{kk}}{E\left[\frac{1}{M \left[\mathbf{W}^{-1}\right]_{kk}}\right]^4} + o\left(\frac{1}{M}\right) \notag
  \\&= O\left(\frac{1}{M}\right), \ M \to \infty. \notag
\end{align}
\end{appendices}

\bibliographystyle{IEEEtran}
\bibliography{IEEEabrv,latencyReTx}

\end{document}